\newcommand{\paperOption}{%
 a4paper,onecolumn,12pt,twoside%
}
\newcommand{\geometryOption}{%
 hscale=0.8,vscale=0.8,centering%
}
\newcommand{\hyperrefOption}{%
 pagebackref,pdfstartview=FitH,hidelinks%
}
\InputIfFileExists{option.aux}{%
 \typeout{Compiling options overridden!}%
}{}

\documentclass[\paperOption]{article}
\usepackage[\geometryOption]{geometry}
\usepackage[\hyperrefOption]{hyperref}

\usepackage{amsmath}
\usepackage{amssymb,amsthm}
\usepackage{bm}
\usepackage{mathenv}
\usepackage{algorithmicx,algpseudocode}
\usepackage[noadjust]{cite}
\usepackage{enumitem}
\usepackage{graphicx}

\newcommand{\docVersion}{?.?.?}
\newcommand{\docBuildNumber}{?}
\InputIfFileExists{version.tex}{
 \typeout{Version set!}
}{}
\InputIfFileExists{bn.aux}{
 \typeout{Build number set!}
}{}

\newenvironment{keywords}%
 {\begin{quote}\textbf{Index Terms} ---}%
 {\relax\end{quote}}

\newtheorem{theorem}{Theorem}[section]
\newtheorem{algorithm}[theorem]{Algorithm}

\newtheorem{conjecture}[theorem]{Conjecture}
\newtheorem{definition}[theorem]{Definition}
\newtheorem{example}[theorem]{Example}

\newtheorem{proposition}[theorem]{Proposition}
\newtheorem{remark}[theorem]{Remark}

\newcommand{\eqdef}{:=}
\newcommand{\relvar}[2]{\buildrel#1\over#2}
\newcommand{\eqvar}[1]{\relvar{\mathrm{#1}}{=}}
\newcommand{\fromTo}[2]{[\![#1,#2]\!]}
\newcommand{\real}{\mathbf{R}}
\newcommand{\set}[2]{{\left\{#1\colon#2\right\}}}
\newcommand{\size}[1]{{\left|#1\right|}}
\newcommand{\transpose}[1]{#1^\mathsf{T}}
\newcommand{\vt}[1]{\bm{#1}}
\DeclareMathOperator{\convexHull}{conv}
\DeclareMathOperator{\rank}{rank}
\DeclareMathOperator{\support}{supp}
\DeclareMathOperator{\weight}{wt}

\newcommand{\canonicalBasis}{\mathrm{e}}
\newcommand{\ceil}[1]{{\left\lceil#1\right\rceil}}
\newcommand{\channel}{\mathcal{W}}
\newcommand{\deterministic}{\mathcal{D}}
\newcommand{\divergence}[2]{D{\left(#1\|#2\right)}}
\newcommand{\floor}[1]{{\left\lfloor#1\right\rfloor}}
\newcommand{\ellone}[1]{\left\|#1\right\|_1}
\newcommand{\map}[3]{#1\colon#2\to#3}
\newcommand{\probability}{\mathcal{P}}
\newcommand{\useless}{\mathrm{U}}
\DeclareMathOperator{\capacity}{C}
\DeclareMathOperator{\decompose}{dec}
\DeclareMathOperator{\diagonal}{diag}
\DeclareMathOperator{\channelDistance}{d}
\DeclareMathOperator{\statDistance}{d}
\DeclareMathOperator{\intrinsic}{IC}
\DeclareMathOperator{\lowerIC}{\underline{\intrinsic}}
\DeclareMathOperator{\upperIC}{\overline{\intrinsic}}
\DeclareMathOperator{\metric}{d}
\DeclareMathOperator{\productMetric}{d_{\vee}}
\DeclareMathOperator{\rankProbability}{\Gamma}
\DeclareMathOperator{\lowerRP}{\underline{\rankProbability}}
\DeclareMathOperator{\upperRP}{\overline{\rankProbability}}

\newcommand{\encoder}{\mathrm{E}}
\newcommand{\decoder}{\mathrm{D}}

\begin{document}

\title{Intrinsic Capacity}
\author{
 Shengtian Yang, Rui Xu, Jun Chen, Jian-Kang Zhang%
 \footnotetext{%
  This work was supported in part by the National Natural Science
  Foundation of China under Grant 61571398 and in part by the
Natural Science and Engineering Research Council (NSERC) of Canada
under a Discovery Grant.
  This paper is to be presented in part at the 2017 IEEE International Symposium on Information Theory.
 }
 \footnotetext{%
  S.~Yang is with the School of Information and Electronic
  Engineering, Zhejiang Gongshang University, Hangzhou 310018, China,
  and was also with the Department of Electrical and Computer
  Engineering, McMaster University, Hamilton, ON L8S 4K1, Canada
  (e-mail: \mbox{yangst@codlab.net}).
 }
 \footnotetext{%
  R.~Xu, J.~Chen, and J.-K. Zhang are with the Department of
  Electrical and Computer Engineering, McMaster University, Hamilton,
  ON L8S 4K1, Canada
  (e-mail: \mbox{xur27@mcmaster.ca}; \mbox{junchen@ece.mcmaster.ca}
  \mbox{jkzhang@ece.mcmaster.ca}).
 }
}
\date{}

\maketitle

\begin{abstract}
Every channel can be expressed as a convex combination of deterministic channels with each deterministic channel
corresponding to one particular intrinsic state. Such convex combinations are in general not unique, each giving
rise to a specific intrinsic-state distribution. In this paper we study the maximum and the minimum capacities of
a channel when the realization of its intrinsic state is causally available at the encoder and/or the decoder.
Several conclusive results are obtained for binary-input channels and binary-output channels. Byproducts of our
investigation include a generalization of  the Birkhoff-von Neumann theorem and a condition on the uselessness of causal state
information at the encoder.
\end{abstract}


\section{Introduction}

A discrete channel is commonly viewed as a black box with the input-output relation characterized by a stochastic matrix.
In practice, it is often possible to obtain some additional information (known as the state information) by probing the channel. The knowledge
of the state information might be useful in increasing the channel capacity. Note that, given each state, the channel can again be viewed as
a black box and can potentially be further probed. One may continue this process until the black box is fully opened, i.e., the channel becomes deterministic
given the acquired state information. This line of thought suggests that every channel has its own intrinsic state, which fully captures the randomness of the channel,
and any state information acquired via channel probing is a degenerate version of this intrinsic state. As such, the intrinsic capacity, defined as the capacity of a channel
when its intrinsic state is revealed, determines the ultimate capacity gain one can hope for by probing the channel.

It turns out that the intrinsic capacity of a channel is not necessarily uniquely defined. Consider a binary symmetric channel with crossover probability
$0.5$:
\(
W
= (W_{x,y})_{x\in\{0,1\},y\in\{0,1\}}
= (\begin{smallmatrix}
0.5 &0.5\\
0.5 &0.5
\end{smallmatrix}),
\)
where each entry $W_{x,y}$ denoting the conditional probability
$W(y\mid x)$ of output $y$ given input $x$.
The capacity of $W$ is clearly zero.
For this channel, we consider the following two models:
\[
F(x) = x\oplus N
\mbox{ and }
G(x) = N,
\]
where $\oplus$ denotes the modulo-$2$ addition and $N$ is
uniformly distributed over $\{0,1\}$.
It is easy to verify that they both have the conditional probability
distribution $W$.
If the actual model of $W$ is $F$, then for every realization of $N$,
$W$ becomes a deterministic perfect channel,
\(
(\begin{smallmatrix} 1 &0\\ 0 &1\end{smallmatrix})
\)
or
\(
(\begin{smallmatrix} 0 &1\\ 1 &0\end{smallmatrix}),
\)
so that the capacity of $W$ with  $N$ available at
the encoder and/or the decoder increases to one.
On the other hand, if the actual model of $W$ is $G$, then for
every realization of $N$, $W$ becomes a deterministic useless channel,
\(
(\begin{smallmatrix} 1 &0\\ 1 &0\end{smallmatrix})
\)
or
\(
(\begin{smallmatrix} 0 &1\\ 0 &1\end{smallmatrix}),
\)
and hence, even with $N$ known at both sides, the capacity of $W$ is
still zero. In fact, it will be seen that, for every number $r\in[0,1]$, one can find a
model for $W$ such that the resulting intrinsic capacity is $r$.

This example indicates that a channel may admit different decompositions into deterministic channels. All these decompositions are mathematically legitimate though the actual
way the deterministic channels are mixed to produce the given channel depends on the underlying physical mechanism. In this work we study the minimum and the maximum intrinsic capacities of a channel over all admissible decompositions. They will be referred to as the lower intrinsic capacity and the upper
intrinsic capacity. For the aforementioned channel $W$, its lower and upper intrinsic capacities are 0 and 1, respectively.
Since the causal state information may be available at the encoder,
the decoder, or both, there are totally three different notions of lower and upper
intrinsic capacities of a channel $W$, denoted by $\lowerIC_f(W)$ and
$\upperIC_f(W)$, for $f=10, 01, 11$, where the two bits indicate if
the state information is available at the encoder and the decoder,
respectively.

The main contributions of this work are:

1)
We study the structure of the convex polytope $\decompose(W)$, which
consists of all convex combinations of deterministic channels for
channel $W$, with a particular focus on its vertices.
It is shown that $\upperIC_f(W)$ for all $f\in\{10,01,11\}$ and $\lowerIC_{11}(W)$ are
attained at certain vertices of $\decompose(W)$
(Theorem~\ref{convexDec}).

2)
We prove a generalization of the Birkhoff-von Neumann theorem for a
family $\channel[a,b]$ of channel matrices with integer-valued
column-sum vector constraints $a$ and $b$ from below and above,
respectively (Theorem~\ref{birkhoff}).
It is shown that $\channel[a,b]$ is convex and its vertices are
exactly all deterministic channels in $\channel[a,b]$.
Using this fundamental result, we determine the exact values of
$\lowerIC_{11}(W)$ and $\upperIC_{11}(W)$ when the input or the output
is binary.
General lower and upper bounds are further provided for the nonbinary
cases (Theorems~\ref{lowerICBound} and \ref{upperICBound}), and in
some cases, the exact value of $\upperIC_{11}(W)$ is also determined.

3)
We obtain the exact values of
$\lowerIC_{10}(W)$ and $\upperIC_{10}(W)$ when $W$ is a binary-output channel (Theorem~\ref{IC10}),
and obtain the exact values of
$\lowerIC_{01}(W)$ and $\upperIC_{01}(W)$ (Proposition~\ref{IC01}) when $W$ is a binary-input channel.
An interesting phenomenon observed is that
$\lowerIC_{10}(W)=\capacity(W)$ for binary-output $W$, where
$\capacity(W)$ denotes the capacity of $W$.
In other words, every binary-output channel can be generated through a certain mechanism such that the capacity remains the same if the source of randomness is
causally revealed to the encoder.
We further prove that the causal state information at the encoder is useless for a broad class of channels
(Theorem~\ref{stateInfoUseless}).
Finally, by providing some counterexamples, we show that the results
such as
$\lowerIC_{10}(W)=\capacity(W)$ and
$\upperIC_{01}(W)=\upperIC_{11}(W)$ are specific to binary-input or binary-output channels, and do not hold in general
(Example~\ref{counterExampleOfIC10} and
Proposition~\ref{counterExampleOfIC01}).

The rest of this paper is organized as follows.
Section~\ref{notations} lists some common notations used
throughout this paper.
Section~\ref{problem} provides the definitions of various notions of (lower/upper) intrinsic capacity and a summary of the main results of this paper.
The proofs and some other relevant findings are
presented in Section~\ref{general} and the appendices.

\section{Notations}\label{notations}

Although most notations will be defined at their first occurrences,
some common ones are listed here for easy
reference.

\begin{description}%
[align=right,parsep=0ex,
 labelwidth=4em,labelsep=0.5em,leftmargin=4.5em]

\item[$\fromTo{x}{y}$]
The set of integers in the interval $[x,y]$.

\item[$B^A$]
The set of all maps $\map{f}{A}{B}$, or equivalently, the set of all
indexed families $x=(x_i\in B)_{i\in A}$ (a generalized form of
sequences).
If $A=\fromTo{1}{n}$, then $B^A$ degenerates to the Cartesian product
$B^n$.
In this paper, a vector (for example, in $\real^n$) will be regarded
as a row vector, and an all-$c$ vector is usually denoted by $\vt{c}$.

\item[$x\wedge y$]
The minimum of $x$ and $y$.

\item[$x\vee y$]
The maximum of $x$ and $y$.

\item[$\support(x)$]
The support set $\set{i\in I}{x_i\ne 0}$ of $x=(x_i)_{i\in I}$.

\item[$\weight(x)$]
The weight $\size{\support(x)}$ of $x=(x_i)_{i\in I}$.

\item[$\floor{x}$]
The largest integer $\le x$.
If the argument is a sequence $x=(x_i)_{i\in I}$, then
$\floor{x}\eqdef (\floor{x_i})_{i\in I}$.
The same convention also applies to other functions such as $|x|$,
$\ceil{x}$, $(x)_+$, and $(x)_-$.

\item[$\ceil{x}$]
The smallest integer $\ge x$.

\item[$(x)_+$]
$x\vee 0$.

\item[$(x)_-$]
$x\wedge 0$.

\item[$\log x$]
$\log_2 x$.
\end{description}

\section{Definitions and Main Results}\label{problem}

Let $\mathcal{X}$ and $\mathcal{Y}$ be two finite sets.
A channel $\map{W}{\mathcal{X}}{\mathcal{Y}}$ is a stochastic matrix
with each entry $W_{x,y}$, or conventionally, $W(y\mid x)$ denoting
the probability of output $y\in\mathcal{Y}$ given input
$x\in\mathcal{X}$.
A deterministic channel $\map{D}{\mathcal{X}}{\mathcal{Y}}$ is a
special channel whose stochastic matrix is a zero-one matrix, as such
it uniquely identifies a map of $\mathcal{X}$ into $\mathcal{Y}$.
In the sequel, deterministic channels and maps will be regarded as
equivalent objects and denoted using the same notation.

It is clear that the set of all channels forms a convex polytope in
$\real^{\mathcal{X}\times\mathcal{Y}}$.
We denote this polytope by $\channel_{\mathcal{X},\mathcal{Y}}$, or
succinctly, $\channel$.
The deterministic channels are exactly the vertices of $\channel$, and every channel can be expressed as a convex combination of
them. This simple observation suggests that, for any channel, one can define a random state variable (referred to as the intrinsic state) given which the channel becomes deterministic.
We are interested in characterizing the capacity of a channel when its intrinsic state is available at the encoder and/or the decoder. Such capacity results are of fundamental importance
since they delineate the potential gain that can be achieved by probing the channel.

For a given channel, there are often multiple ways to write it as a convex combination of deterministic channels; as a consequence, the distribution of its intrinsic state is in general not uniquely defined.
Let $\deterministic_{\mathcal{X},\mathcal{Y}}$ (or simply $\deterministic$) denote the set of all deterministic channels $\mathcal{X}\to\mathcal{Y}$.
Then the set of all possible convex decompositions of a channel $W$ is given by
\[
\decompose(W)
\eqdef \set{\lambda\in\probability_\deterministic}%
 {W=\sum_{D\in\deterministic} \lambda_D D},
\]
where $\probability_\deterministic$ is the set of all probability
distributions over $\deterministic$ and can be regarded as the set
$\channel_{\{\emptyset\},\deterministic}$ of matrices or vectors.
For each intrinsic-state distribution $\lambda\in\probability_\deterministic$, we define the resulting capacities when the intrinsic state is causally available at the encoder, the
decoder, or both, by
\begin{eqnarray*}
\capacity_{10}(\lambda)
&\eqdef &\max_{\mu\in\probability_{\mathcal{X}^{\deterministic}}}
 J_{10}(\lambda,\mu)\\
&= &\max_{\mu\in\probability_{\mathcal{X}^{\deterministic}}}
 I{\left(\mu, \left(\sum_{D}
 \lambda_D D_{u(D),y}\right)_{u\in\mathcal{X}^{\deterministic},
 y\in\mathcal{Y}}\right)},\\
\capacity_{01}(\lambda)
&\eqdef &\max_{\mu\in\probability_\mathcal{X}} J_{01}(\lambda,\mu)
= \max_{\mu\in\probability_\mathcal{X}}
 \sum_{D} \lambda_D I(\mu,D),\\
\capacity_{11}(\lambda)
&\eqdef &\max_{\kappa\in\channel_{\deterministic,\mathcal{X}}}
 \sum_{D} \lambda_D I(\kappa_{D,*},D)
= \sum_{D} \lambda_D \log\rank(D),
\end{eqnarray*}
respectively (see \cite[Chapter~7]{el_gamal_network_2011}), where
\[
I(\mu,W)
\eqdef \sum_{x} \mu_x \divergence{W_{x,*}}{\mu W}
\]
and the flag $f\in\{10,01,11\}$ indicates the availability of the intrinsic state at
the encoder and the decoder.
For example, $10$ means that the intrinsic state is available at the
encoder but not at the decoder.
For completeness, we also define the capacity with no encoder and decoder side
information:
\[
\capacity(W)
= \capacity_{00}(\lambda)
\eqdef \max_{\mu\in\probability_\mathcal{X}} I{\left(\mu,
 \sum_{D\in\deterministic} \lambda_D D\right)}
= \max_{\mu\in\probability_\mathcal{X}} I(\mu,W).
\]
Then, given a channel $W$, we can define its intrinsic-capacity set
by
\[
\intrinsic_f(W)
\eqdef \set{\capacity_f(\lambda)}{\lambda\in\decompose(W)}.
\]
Furthermore, we define the lower intrinsic capacity and the upper
intrinsic capacity of $W$ for $f\in\{10,01,11\}$  by
\[
\lowerIC_f(W)
\eqdef \inf \intrinsic_f(W)
\]
and
\[
\upperIC_f(W)
\eqdef \sup \intrinsic_f(W),
\]
respectively.

\begin{remark}
Using the functional representation lemma
\textnormal{\cite[p.~626]{el_gamal_network_2011}\cite[Lemma~1]{wczcp11}}, it can be easily
shown  that $\upperIC_f(W)$ provides an upper
bound on the capacity of $W$ with any form of state information whose availability at
the encoder and the decoder is specified by $f$.
On the other hand, from the minimax theorem
\textnormal{\cite{nikaido_von_1954}},
Proposition~\ref{propertyOfJ}, and
\textnormal{\cite[Theorems~7.1 and 7.2, Eqs.~(7.2) and (7.3), and
 Remark~7.6]{el_gamal_network_2011}},
it follows that $\lowerIC_f(W)$ is exactly the capacity of the compound channel
$(S)_{p_S\in\decompose(W)}$ with the availability of $S$ at the encoder and
the decoder specified by $f$, where $S$ is $\deterministic$-valued, i.e., a random
deterministic channel, and $p_S$ is selected arbitrarily from
$\decompose(W)$.
\end{remark}

The main results of this paper are given as follows.
With no loss of generality, we assume from now on that the channel $W$
is from $\fromTo{1}{m}$ to $\fromTo{1}{n}$, where $m,n\ge 2$.

\begin{definition}\label{rankP}
Let
\[
\rankProbability_\lambda(r)
\eqdef \lambda\set{D\in\deterministic}{\rank(D)=r}
\]
be the rank probability function over $\deterministic$ induced by
$\lambda\in\decompose(W)$.
The lower and the upper rank-$r$ probabilities of $W$ are then defined
by
\[
\lowerRP_W(r)
\eqdef \min_{\lambda\in\decompose(W)}\rankProbability_\lambda(r)
\ \mbox{and}\
\upperRP_W(r)
\eqdef \max_{\lambda\in\decompose(W)}\rankProbability_\lambda(r),
\]
respectively.
\end{definition}

Bounds for $\lowerRP_W(r)$ and $\upperRP_W(r)$ when $r=1$ and
$r=m\wedge n$ are given by Propositions~\ref{useless} and
\ref{perfect}, respectively.
Most of our results will be expressed in terms of these quantities.

\begin{theorem}\label{lowerICBound}
\[
\lowerIC_{11}(W)
\le \begin{cases}
(1-\upperRP_W(1))\log\gamma
 &\text{$\upperRP_W(1)<1$,}\\
0 &\text{otherwise,}
\end{cases}
\]
\[
\lowerIC_{11}(W)\ge 1-\upperRP_W(1),
\]
where
\[
\upperRP_W(1)
= \sum_{j=1}^n \alpha_j,
\
\alpha=\left(\min_{1\le i\le m}W_{i,j}\right)_{j\in\fromTo{1}{n}},
\]
\[
\gamma = (m+\weight(a)-a\transpose{\vt{1}})\wedge n,
\]
\[
a=\floor{\vt{1}W'},
\
W'=\frac{W-\sum_{j=1}^n\alpha_j\useless_j}{1-\upperRP_W(1)},
\]
and $\useless_j$ is the deterministic useless channel with its $j$-th
column being all one.

If $m=2$ or $n=2$, then $\lowerIC_{11}(W)=1-\upperRP_W(1)$.
\end{theorem}

\begin{theorem}\label{upperICBound}
If $\lowerRP_W(1)>0$ or $m=2$ or $n=2$, then
\[
\upperIC_{11}(W)=1-\lowerRP_W(1);
\]
otherwise,
\[
\log\gamma
\le \upperIC_{11}(W)
\le \log(o-1)+\upperRP_W(o)\log\frac{o}{o-1},
\]
where
\[
\lowerRP_W(1)
= \left(g - m + 1\right)_+,
\
g = \max_{1\le j\le n} (\vt{1}W)_j,
\]
\[
\gamma
= \weight(a)+\left(m-\sum_{j\in\support(a)}b_j\right)_+,
\
o = m\wedge n,
\]
\[
a=\floor{\vt{1}W},
\
b=\ceil{\vt{1}W}.
\]

If $m\le n$ and $\vt{1}W\le \vt{1}$, then $\upperIC_{11}(W)=\log m$.

If $m\ge n$ and $\vt{1}W\ge \vt{1}$, then $\upperIC_{11}(W)=\log n$.
\end{theorem}

\begin{theorem}\label{IC10}
If $n=2$, then
\[
\lowerIC_{10}(W)=\capacity(W)
\]
and
\[
\upperIC_{10}(W)=\capacity\left(\begin{pmatrix}
 1 &0\\
 \lowerRP_W(1) &1-\lowerRP_W(1)
 \end{pmatrix}\right).
\]
\end{theorem}

\begin{proposition}\label{IC01}
If $m=2$, then for every $\lambda\in\decompose(W)$,
$\capacity_{01}(\lambda)=\capacity_{11}(\lambda)$, so that
$\lowerIC_{01}(W)=1-\upperRP_W(1)$ and
$\upperIC_{01}(W)=1-\lowerRP_W(1)$.
\end{proposition}


The above results enable us to obtain explicit characterizations of all lower and upper intrinsic capacities for binary-input binary-output channels. The relevant expressions are collected in the following example.


\begin{example}
If $m=n=2$ and
\[
W
= \begin{pmatrix}
 1-\epsilon_1 &\epsilon_1\\
 \epsilon_2 &1-\epsilon_2
 \end{pmatrix},
\]
then
\[
\lowerIC_{11}(W)
= \lowerIC_{01}(W)
= |1-\epsilon_1-\epsilon_2|,
\]
\[
\lowerIC_{10}(W)=\capacity(W)=\begin{cases}
\log\left(
 2^{\frac{\epsilon_2 h(\epsilon_1) - (1-\epsilon_1) h(\epsilon_2)}
  {1-\epsilon_1-\epsilon_2}}
 + 2^{\frac{\epsilon_1 h(\epsilon_2) - (1-\epsilon_2) h(\epsilon_1)}
  {1-\epsilon_1-\epsilon_2}}
\right),  &$\epsilon_1+\epsilon_2\neq 1$,\\
0 &$\epsilon_1+\epsilon_2=1$,
\end{cases}
\]
\[
\upperIC_{11}(W)
= \upperIC_{01}(W)
= 1-|\epsilon_1-\epsilon_2|,
\]
\[
\upperIC_{10}(W)
=  \begin{cases}
1 &$\epsilon_1=\epsilon_2$,\\
\log\left(1+(1-|\epsilon_1-\epsilon_2|)|\epsilon_1-\epsilon_2|^{\frac{|\epsilon_1-\epsilon_2|}{1-|\epsilon_1-\epsilon_2|}}\right)
 &$|\epsilon_1-\epsilon_2|\in (0,1)$,\\
0, &$|\epsilon_1-\epsilon_2|=1$,
\end{cases}
\]
where
$h(\epsilon)\eqdef-\epsilon\log\epsilon-(1-\epsilon)\log(1-\epsilon)$
is the binary entropy function.
If $W$ is a binary symmetric
channel with crossover probability $\epsilon$ (i.e., $\epsilon_1=\epsilon_2=\epsilon$), then
\[
\lowerIC_{11}(W)
= \lowerIC_{01}(W)
= |1-2\epsilon|,
\]
\[
\lowerIC_{10}(W)
=C(W)
= 1-h(\epsilon),
\]
\[
\upperIC_{11}(W)
= \upperIC_{01}(W)
= \upperIC_{10}(W)
= 1.
\]
If $W$ is a Z-channel
with crossover probability $\theta$ (i.e., $\epsilon_1=0$ and $\epsilon_2=\theta$), then
\[
\lowerIC_{11}(W)
= \lowerIC_{01}(W)
= \upperIC_{11}(W)
= \upperIC_{01}(W)
= 1-\theta,
\]
\[
\lowerIC_{10}(W)
= \upperIC_{10}(W)
=C(W)
= \begin{cases}
1 &$\theta=0$,\\
\log\left(1+(1-\theta)\theta^{\frac{\theta}{1-\theta}}\right)
 &$\theta\in (0,1)$,\\
0, &$\theta=1$.
\end{cases}
\]
The case of Z-channel is special, because in this case $W$ admits a unique convex decomposition into deterministic channels:
\[
\theta\begin{pmatrix}1 &0\\ 1 &0\end{pmatrix}
+ (1-\theta)\begin{pmatrix}1 &0\\ 0 &1\end{pmatrix}.
\]
The lower and the upper intrinsic capacities of these two special channels are plotted in
Figs.~\ref{fig:cap.bsc} and \ref{fig:cap.zc}.
\begin{figure}[h!]
 \centering
 \includegraphics[width=10cm]{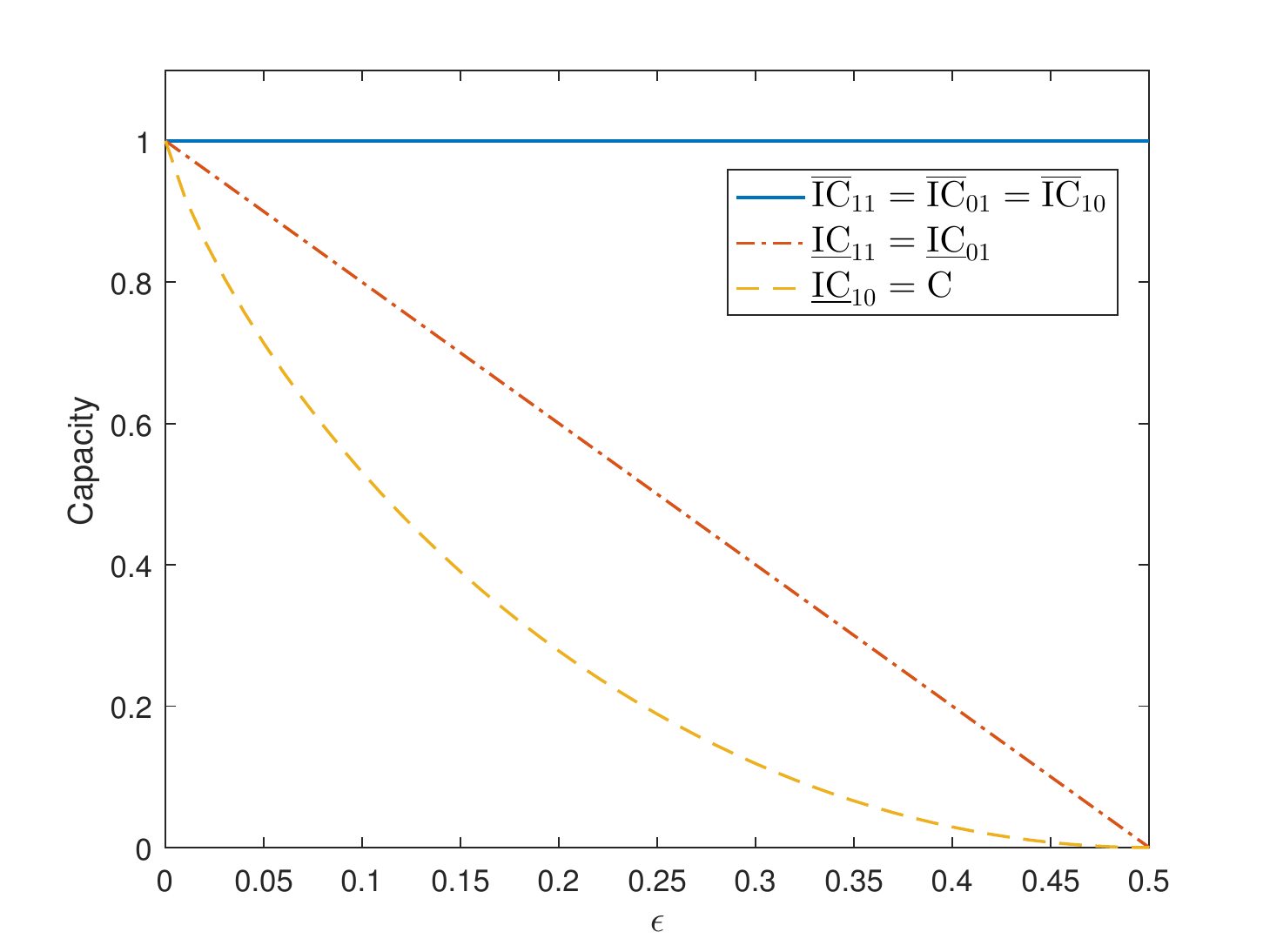}
 \caption{The lower and the upper intrinsic capacities of a binary symmetric channel with
  crossover probability $\epsilon$.}
 \label{fig:cap.bsc}
\end{figure}
\begin{figure}[h!]
 \centering
 \includegraphics[width=10cm]{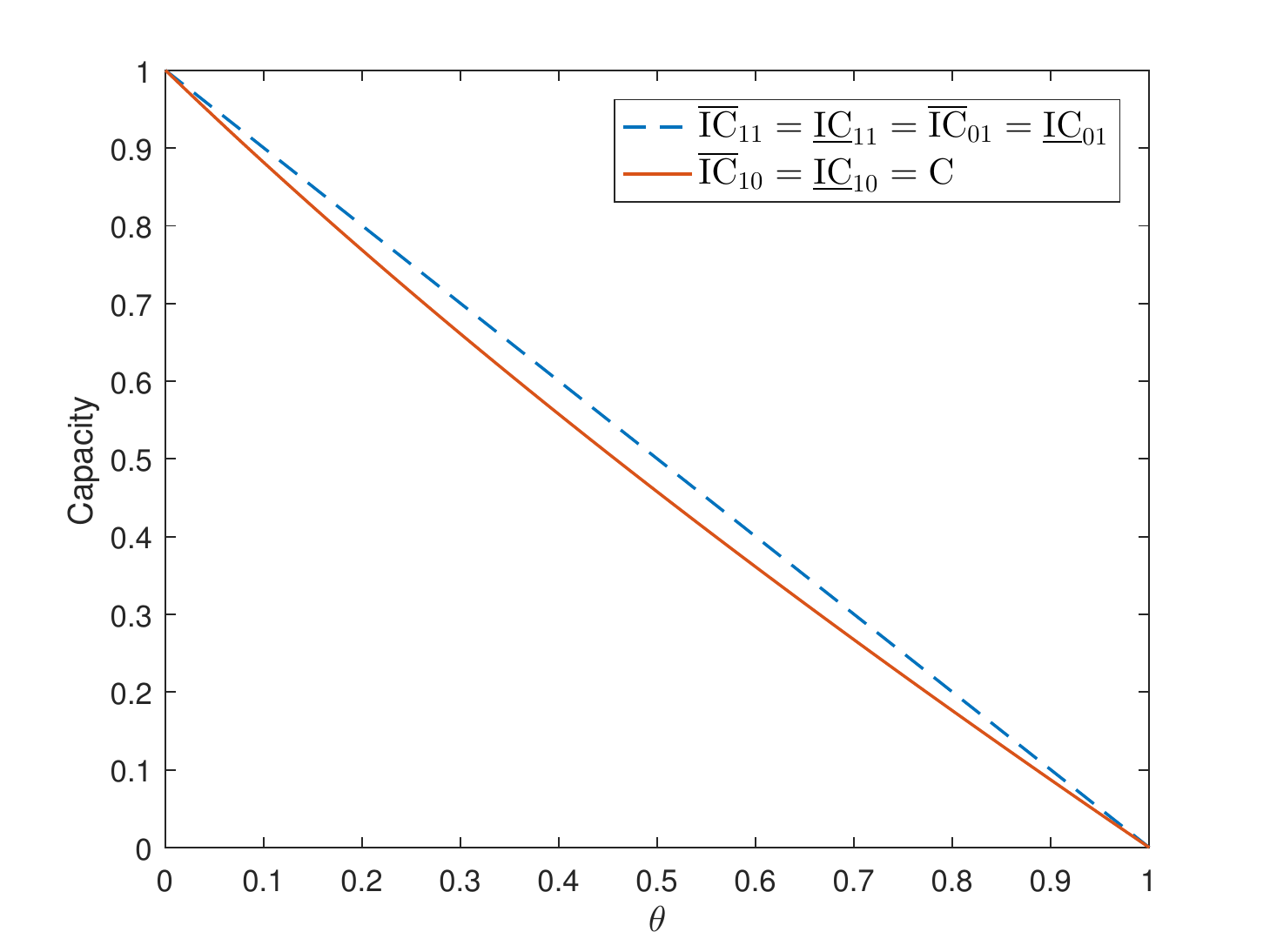}
 \caption{The lower and the upper intrinsic capacities of a Z-channel with crossover
  probability $\theta$.}
 \label{fig:cap.zc}
\end{figure}
\end{example}

\section{Proofs of Main Results}\label{general}

It is clear that $\decompose(W)$ is bounded, closed, and convex, so it
can be  easily shown that $\intrinsic_f(W)$ is a closed interval and
that $\upperIC_f(W)$ for all $f\in\{10,01,11\}$ and $\lowerIC_{11}(W)$ are attained at certain
vertices of $\decompose(W)$ (Theorem~\ref{convexDec}).
As such, it is of great importance to study the structure of $\decompose(W)$.
A series of results on the vertices of $\decompose(W)$ is provided in
Appendix~\ref{structureOfDec}. Although these results shed useful light on the structure of
$\decompose(W)$, the characterizations are still too coarse for our purpose. It will be seen that additional
insights can be gained by taking the objective functions into consideration.


\subsection{$\lowerIC_{11}(W)$ and $\upperIC_{11}(W)$}
\label{secIC11}

We first provide a complete characterization of $\lambda$ that
achieves $\lowerIC_{11}(W)$ or $\upperIC_{11}(W)$.

\begin{proposition}\label{characterizationOfIC}
Let
\[
\mathfrak{U}_+
= \set{\support((\alpha)_+)}{\alpha\in\real^\deterministic,
 \sum_{D\in\deterministic} \alpha_D D = 0,
 \sum_{D\in\deterministic} \alpha_D \log\rank(D) > 0}
\]
and
\[
\mathfrak{U}_-
= \set{\support((\alpha)_-)}{\alpha\in\real^\deterministic,
 \sum_{D\in\deterministic} \alpha_D D = 0,
 \sum_{D\in\deterministic} \alpha_D \log\rank(D) > 0}.
\]
For $\lambda\in\decompose(W)$, $\capacity_{11}(\lambda)=\lowerIC_{11}(W)$
iff there is no $U\in\mathfrak{U}_+$ such that
$U\subseteq\support(\lambda)$;
$\capacity_{11}(\lambda)=\upperIC_{11}(W)$ iff there is no
$U\in\mathfrak{U}_-$ such that $U\subseteq\support(\lambda)$.
\end{proposition}

\begin{proof}
It suffices to prove the first part, because the second part can be
proved in the same vein.

(Sufficiency) If there exists some $\beta\in\decompose(W)$ such that
$\capacity_{11}(\beta)<\capacity_{11}(\lambda)$, then
\[
\sum_{D\in\deterministic} (\lambda_D-\beta_D) D = 0
\]
and
\[
\sum_{D\in\deterministic} (\lambda_D-\beta_D) \log\rank(D) > 0,
\]
so that $U=\support((\lambda-\beta)_+)\in\mathfrak{U}_+$ and
$U\subseteq\support(\lambda)$, a contradiction.

(Necessity) For $U\in\mathfrak{U}_+$, if
$U\subseteq\support(\lambda)$, then there is a vector
$\alpha\in\real^\deterministic$ such that
$\support((\alpha)_+)\subseteq\support(\lambda)$,
$\sum_D \alpha_D D=0$, and $\sum_D \alpha_D\log\rank(D)>0$.
Let $\beta=\lambda-t\alpha$.
For sufficiently small $t>0$, it can be verified that
$\beta\in\decompose(W)$ and
$\capacity_{11}(\beta)<\capacity_{11}(\lambda)=\lowerIC_{11}(W)$, which is
absurd.
\end{proof}

\begin{definition}
A subset $S\subseteq\deterministic$ is said to be
$\intrinsic_{11}$-minimized, or succinctly, $\intrinsic$-minimized
(resp., $\intrinsic$-maximized) if there is a
$\lambda\in\probability_\deterministic$ such that
$\support(\lambda)=S$ and $\capacity_{11}(\lambda)=\lowerIC_{11}(W)$
(resp., $\capacity_{11}(\lambda)=\upperIC_{11}(W)$), where
$W=\sum_{D\in\deterministic} \lambda_D D$.
\end{definition}

A simple consequence of Proposition~\ref{characterizationOfIC} is:

\begin{proposition}\label{necessityOfIC}
If $S\subseteq\deterministic$ is $\intrinsic$-minimized (resp.,
$\intrinsic$-maximized), then any
$\lambda\in\probability_\deterministic$ supported on $S$ achieves
$\lowerIC_{11}(W)$ (resp., $\upperIC_{11}(W)$), where
$W=\sum_{D\in\deterministic} \lambda_D D$.
As a consequence, any nonempty subset of $S$ is also
$\intrinsic$-minimized (resp., $\intrinsic$-maximized).
\end{proposition}

By Proposition~\ref{necessityOfIC}, it is important to identify
patterns of sets that are not $\intrinsic$-minimized or
$\intrinsic$-maximized.
Some simple patterns that are not $\intrinsic$-minimized or
$\intrinsic$-maximized are given as follows and their proofs are
relegated to Appendix~\ref{secIC11extra}.

\begin{proposition}\label{min.pattern.1}
If $m\le n$, then any deterministic perfect channels $P_1$, \ldots,
$P_\ell$ such that at least one column of $P_1+\cdots+P_\ell$ has a
weight greater than one are not $\intrinsic$-minimized.
\end{proposition}

\begin{proposition}\label{min.pattern.2}
If $m\ge n$, then any deterministic perfect channels $P_1$, \ldots,
$P_\ell$ such that at least one column of $P_1+\cdots+P_\ell$ has no
entry equal to $\ell$ are not $\intrinsic$-minimized.
\end{proposition}

\begin{proposition}\label{max.pattern.1}
For $D\in\deterministic$, if $\weight(D_{*,j})\le m-2$, then
$\{D, \useless_j\}$ is not $\intrinsic$-maximized.
\end{proposition}

The next result is a generalization of the Birkhoff-von Neumann
theorem, which plays a crucial role in proving Theorems~\ref{lowerICBound} and
\ref{upperICBound}.
Our proof hinges on an extension of the ideas in
\cite{jurkat_extremal_1968,caron_nonsquare_1996}.

\begin{theorem}\label{birkhoff}
Let $a$ and $b$ be two $n$-dimensional integer-valued vectors such
that $a\le b$, namely, $a_j\le b_j$ for $1\le j\le n$.
Let
\[
\channel[a,b]
\eqdef \set{W\in\channel}{a\le \vt{1}W\le b}
\]
and $\deterministic[a,b]\eqdef \channel[a,b]\cap\deterministic$, where
$\vt{1}$ denotes the $m$-dimensional all-one row vector.
If $\channel[a,b]$ is not empty, then $\channel[a,b]$ is convex and
the vertices of $\channel[a,b]$ are exactly the matrices in
$\deterministic[a,b]$.
\end{theorem}

\begin{proof}
It is clear that $\channel[a,b]$, if nonempty, is a convex set.
We will show that any matrix $W\in\channel[a,b]$ with non-integer
entries cannot be a vertex of $\channel[a,b]$.
There are two cases:

Case (a):
There is a non-integer entry in a non-boundary column.

Case (b):
All non-integer entries are in the boundary columns.

\noindent
Here, a column is
called a boundary column if its sum is either $a_j$ or
$b_j$, where $j$ is the index of the column.

In whichever the case, we can pick a non-integer entry, say the
$(i_0,j_0)$ entry, which in Case (a) must be a non-integer entry in
a non-boundary column.
By the following argument, we will find a chain or loop of non-integer
entries of the matrix, which will be used to prove that the matrix is
not extremal.

Because the $(i_0,j_0)$ entry is not an integer, there exists at least
another entry in the same row that is also not an integer, say the
$(i_0,j_1)$ entry.
If the $j_1$-th column is not on the boundary, then we are done.
If however the $j_1$-th column is on the boundary, then there exists at
least another non-integer entry in the same column, say $(i_1,j_1)$.
In general, after $t$ steps, we have visited $t+1$ columns, with the
chain
\[
(i_0,j_0), (i_0,j_1), (i_1,j_1), \ldots,
(i_{t-1},j_t), (i_t,j_t).
\]
Except for the $j_0$-th column, every column has exactly one inbound
entry $(i_{s-1},j_s)$ and one outbound entry $(i_s,j_s)$, where
$1\le s\le t$.
Now in the $(t+1)$-th step, by the same argument, we find the
$(i_t,j_{t+1})$ entry in the $j_{t+1}$-th column.
If this column has already been visited, then $j_{t+1}=j_s$ for some
$0\le s\le t-1$ and we are done.
If this column is new but not on the boundary, we are also done.
If however this new column is on boundary, then we can further find
an outbound entry in this column, say $(i_{t+1},j_{t+1})$, and
proceed to the $(t+2)$-th step.
Because there are finite columns, we will always end up with a chain
\[
(i_0,j_0), (i_0,j_1), (i_1,j_1), \ldots,
(i_{k-1},j_{k-1}), (i_{k-1},j_k)
\]
which only happens in Case (a), or a loop
\[
(i_\ell,j_\ell), (i_\ell,j_{\ell+1}), (i_{\ell+1},j_{\ell+1})\ldots,
(i_{k-1},j_k), (i_k,j_k)=(i_\ell,j_\ell)
\]
for some $0\le\ell<k-1$.

Then we can construct a matrix $N$ by setting all outbound entries (in
the chain or the loop) $N_{i_s,j_s}=1$, all inbound entries
$N_{i_{s-1},j_s}=-1$, and all other entries to be zero.
It is clear that
\[
\vt{1}N=\canonicalBasis_{j_0}-\canonicalBasis_{j_k},
\quad
N\transpose{\vt{1}}=0
\]
in the former case and
\[
\vt{1}N=0,
\quad
N\transpose{\vt{1}}=0
\]
in the latter case, where
$\canonicalBasis_{k}=(1\{j=k\})_{j\in\fromTo{1}{n}}$.

Let $U=W+\epsilon N$ and $V=W-\epsilon N$.
It is clear that $U,V\in\channel[a,b]$ for sufficiently small
$\epsilon>0$.
It is also clear that $W=\frac{1}{2}U+\frac{1}{2}V$ and $U\ne V$, that
is, $W$ is not a vertex of $\channel[a,b]$.

Therefore, we have $\mathcal{V}\subseteq\deterministic[a,b]$, where $\mathcal{V}$ denotes the set of all vertices of $\channel[a,b]$.
It remains to show that $\deterministic[a,b]\subseteq\mathcal{V}$.
For any $W\in\deterministic[a,b]$, if $W=\alpha U+(1-\alpha)V$ with
$U,V\in\channel[a,b]$ and $\alpha\in(0,1)$, then for every
$1\le i\le m$,
\[
\canonicalBasis_i W
=\alpha \canonicalBasis_i U+(1-\alpha) \canonicalBasis_i V,
\]
which however implies that $\canonicalBasis_i U=\canonicalBasis_i V$
for every $1\le i\le m$, or $U=V$.
\end{proof}

Equipped with Theorem~\ref{birkhoff}, we proceed to derive bounds for the lower
and the upper rank probabilities (Definition~\ref{rankP}). These bounds are useful in estimating the lower and the upper
intrinsic capacities.

\begin{proposition}\label{useless}
\[
\lowerRP_W(1)
= \left(g - m + 1\right)_+,
\]
\[
\upperRP_W(1)
= \sum_{j=1}^n \alpha_j,
\]
where
\begin{equation}\label{maximumColumnSum}
g = \max_{1\le j\le n} (\vt{1}W)_j,
\end{equation}
\begin{equation}\label{uselessMaximal}
\alpha=\left(\min_{1\le i\le m}W_{i,j}\right)_{j\in\fromTo{1}{n}}.
\end{equation}
\end{proposition}

\begin{proof}
By Theorem~\ref{birkhoff}, $W$ can be expressed as a convex
combination of deterministic channels of rank $\ge 2$ if $g\le m-1$,
in which case, $\lowerRP_W(1)=0$.
Otherwise, let $\ell$ be the index of the column with the sum $g>m-1$.
Consider the convex combination
\[
W = t \useless_\ell + (1-t) W'.
\]
It is clear that $W'$ cannot be a convex combination of deterministic
channels of rank $\ge 2$ unless the sum of its $\ell$-th column is
$\le m-1$.
To this end, we set $t=g-m+1$, which is the minimum value
required, and we have
\[
(\vt{1}W')_\ell
= \frac{(\vt{1}W)_\ell-t m}{1-t}
= m-1
\]
and
\[
(\vt{1}W')_j
= \frac{(\vt{1}W)_j}{1-t}
\le 1
\]
for $j\ne\ell$, so that $\lowerRP_W(1)=g-m+1$.

If $W$ has the following convex decomposition
\[
W
= \left(1-\sum_{j=1}^n s_j\right) W'
 + \sum_{j=1}^n s_j \useless_j,
\]
then $W'$ is a valid stochastic matrix iff $s_j\le\alpha_j$ for all
$j$.
Therefore, $\upperRP_W(1) = \sum_{j=1}^n \alpha_j$.
\end{proof}

\begin{proposition}\label{upperICStragtegy}
If $\lambda\in\decompose(W)$ achieves $\upperIC_{11}(W)$, then
$\rankProbability_\lambda(1)=\lowerRP_W(1)$.
In particular, if $\lowerRP_W(1)>0$, then
$\lambda_{\useless_\ell}=\lowerRP_W(1)$ and
$\rankProbability_\lambda(2)=1-\lowerRP_W(1)$, where
$\ell=\arg\max_{1\le j\le n} (\vt{1}W)_j$.
\end{proposition}

\begin{proof}
If $\lambda$ is zero on all deterministic useless channels, then
$\rankProbability_\lambda(1)=\lowerRP_W(1)=0$.

If $\lambda_{\useless_j}>0$ for some $j$, then $\lambda$ must be zero
on all deterministic channels whose $j$-th column weight is less than
$m-1$ (Propositions~\ref{necessityOfIC} and \ref{max.pattern.1}).
Therefore, we must have
$\lambda_{\useless_j}=\rankProbability_\lambda(1)=\lowerRP_W(1)$
(Proposition~\ref{useless}) and
$\rankProbability_\lambda(2)=1-\lowerRP_W(1)$.
\end{proof}

\begin{proposition}\label{perfect}
If $m\le n$, then
\[
\upperRP_W(m)
\le 1-\beta,
\]
where
\begin{equation}\label{perfect1}
\beta = 0\vee\max_{1\le j\le n} \beta'_j
\end{equation}
and
\[
\beta'_j
= \begin{cases}
\frac{(1W)_j-1}{\weight(W_{*,j})-1} &$\weight(W_{*,j})>1$,\\
0 &otherwise.
\end{cases}
\]
Furthermore, if $\beta=0$, then $\upperRP_W(m)=1$.

If $m\ge n$, then
\[
\upperRP_W(n)\le h,
\]
where
\begin{equation}\label{perfect2}
h = 1\wedge\min_{1\le j\le n} (1W)_j.
\end{equation}
If $h=1$, then $\upperRP_W(n)=1$.
\end{proposition}

\begin{proof}
If $m\le n$, then the sum of every column of a deterministic channel
of rank $m$ is at most 1, and for every $1\le j\le n$, $W$ admits a
convex decomposition into deterministic channels with the $j$-th column
sum at most $\weight(W_{*,j})$.
Thus for every $\lambda\in\decompose(W)$ and every $j$,
\begin{eqnarray*}
(\vt{1}W)_j
&\le &(1\wedge\weight(W_{*,j}))\rankProbability_\lambda(m)
 + \weight(W_{*,j})(1-\rankProbability_\lambda(m))\\
&= &\weight(W_{*,j})-(\weight(W_{*,j})-1)_+\rankProbability_\lambda(m),
\end{eqnarray*}
so that
\[
\rankProbability_\lambda(m)
\le 1-\frac{(\vt{1}W)_j-1}{\weight(W_{*,j})-1}
\]
for $\weight(W_{*,j})>1$ and hence $\upperRP_W(m)\le 1-\beta$.
If $\beta=0$, which implies that $(\vt{1}W)_j\le 1$ for all
$1\le j\le n$, then $\upperRP_W(m)=1$ (Theorem~\ref{birkhoff}).

If $m\ge n$, then the sum of every column of a deterministic channel
of rank $n$ is at least $1$, so that, for every
$\lambda\in\decompose(W)$ and every $1\le j\le n$,
\[
(\vt{1}W)_j\ge \rankProbability_\lambda(n),
\]
and hence $\upperRP_W(n)\le h$.
If $h=1$, which implies $(\vt{1}W)_j\ge 1$ for all $1\le j\le n$, then
$\upperRP_W(n)=1$ (Theorem~\ref{birkhoff}).
\end{proof}

We are now ready to prove Theorems~\ref{lowerICBound} and
\ref{upperICBound}.

\begin{proof}[Proof of Theorem~\ref{lowerICBound}]
To find an upper bound of $\lowerIC_{11}(W)$, we need to find a convex
decomposition of $M$ as ``bad" as possible.
To this end, we can first extract from $W$ a collection of useless
channels with the total probability $\upperRP_W(1)$
(Proposition~\ref{useless}), that is,
\[
W = \sum_{j=1}^n\alpha_j\useless_j+(1-\upperRP_W(1))W'.
\]
If $\upperRP_W(1)=1$, then $\lowerIC_{11}(W)=0$; otherwise,
\[
\lowerIC_{11}(W)\le (1-\upperRP_W(1))\lowerIC_{11}(W').
\]

It is clear that $W'\in\channel[a,\vt{m}]$, where $\vt{m}$ denotes the
all-$m$ row vector.
The best deterministic channels in $\channel[a,m]$ are those with the
number of nonzero columns maximized.
The rank of those matrices is
\[
\left(\weight(a)+m-\sum_{j=1}^n a_j\right)\wedge n,
\]
so
\(
\lowerIC_{11}(W')
\le \log((m+\weight(a)-a\transpose{\vt{1}})\wedge n)
\)
(Theorem~\ref{birkhoff}).

Let $\lambda$ be a vertex of $\decompose(W)$ that attains
$\lowerIC_{11}(W)$.
Then
\begin{eqnarray*}
\lowerIC_{11}(W)
&= &\sum_{D\in\deterministic} \lambda_D \log\rank(D)\\
&\ge &1-\rankProbability_\lambda(1)
\ge 1-\upperRP_W(1).
\end{eqnarray*}
Finally, the special case of $m=2$ or $n=2$ can be easily verified.
\end{proof}

\begin{proof}[Proof of Theorem~\ref{upperICBound}]
Let $\lambda$ be a vertex of $\decompose(W)$ that attains
$\upperIC_{11}(W)$.

If $\lowerRP_W(1)>0$ or $m=2$ or $n=2$, then
$\rankProbability_\lambda(r)=0$ for all $r>2$
(Proposition~\ref{upperICStragtegy}), so that
$\upperIC_{11}(W)=1-\lowerRP_W(1)$ (Proposition~\ref{useless}).
The remaining case is then $\lowerRP_W(1)=0$.

To find a lower bound of $\upperIC_{11}(W)$, we need to find a convex
decomposition of $W$ as ``good" as possible.
It is clear that $W\in\channel[a,b]$, so $\upperIC_{11}(W)$ is bounded
below by the capacity of the worst deterministic channel in
$\channel[a,b]$ (Theorem~\ref{birkhoff}), which are obviously those
with the number of nonzero columns minimized.
The capacity of such a channel is $\log\gamma$, so that
$\upperIC_{11}(W)\ge \log\gamma$.

On the other hand,
\begin{eqnarray*}
\upperIC_{11}(W)
&= &\sum_{D\in\deterministic} \lambda_D \log\rank(D)\\
&\le &(1-\rankProbability_\lambda(o))\log(o-1)
 + \rankProbability_\lambda(o)\log o\\
&= &\log(o-1) + \rankProbability_\lambda(o)\log\frac{o}{o-1}\\
&\le &\log(o-1)+\upperRP_W(o)\log\frac{o}{o-1}
\end{eqnarray*}
where $o=m\wedge n$.
The remaining part of the proof is straightforward.
\end{proof}

The bounds given by Theorems~\ref{lowerICBound} and \ref{upperICBound}
can be improved in various ways.
In Theorem~\ref{lowerICBound}, if $\gamma=m\wedge n$, then the upper
bound for $\upperRP_W(m\wedge n)$ in Proposition~\ref{perfect} can be
used to improve the upper bound for $\lowerIC_{11}(W)$; if $\gamma=m=n$,
the upper bound for $\lowerIC_{11}(W)$ can be improved by
Proposition~\ref{min.pattern.1} (see Example~\ref{concrete}).
The lower bound for $\lowerIC_{11}(W)$ can also be improved by
$(1-\upperRP_W(1))\vee\capacity(W)$ because
$\capacity(W)\le\lowerIC_{11}(W)$.
However, all these improvements are somewhat ad hoc.
The fundamental problem to be solved is how we can choose $\lambda$ in
order to approach or achieve the lower or the upper intrinsic
capacities.
In particular, based on Theorems~\ref{lowerICBound}, we have the
following conjecture:

\begin{conjecture}
For $\lambda\in\decompose(W)$, if
$\capacity_{11}(\lambda)=\lowerIC_{11}(W)$, then
$\rankProbability_\lambda(1)=\upperRP_W(1)$.
\end{conjecture}

\subsection{$\lowerIC_{10}(W)$ and $\upperIC_{10}(W)$}
\label{secIC10}

Although it is difficult to compute $\lowerIC_{10}(W)$ and
$\upperIC_{10}(W)$ in general, their exact values can be
determined in the binary-output case, as is shown by
Theorem~\ref{IC10}.

\begin{proof}[Proof of Theorem~\ref{IC10}]
Since $n=2$, we only need to choose two maps from all the
$2^{|\deterministic|}=2^{2^m}$ maps of $\deterministic$ into
$\fromTo{1}{m}$ for constructing the capacity-achieving distributions.
We denote these two maps by $u$ and $v$.
The optimal strategy for choosing $u,v$ is to maximize $W'_{u,1}$ and
minimize $W'_{v,1}$, where
$W'_{u,y}=\sum_{D\in\deterministic} \lambda_D D_{u(D),y}$.
There are only two classes of deterministic channels in
$\deterministic$, rank $1$ and rank $2$.
For $D$ of rank $1$, it does not matter how to choose the values of $u(D)$
and $v(D)$.
For $D$ of rank $2$, however, we choose $u(D)=i_1$ such that
$D_{i_1,1}=1$ and choose $v(D)=i_2$ such that $D_{i_2,1}=0$.
Then we have
\[
W'_{u,*}
= (1-\lambda_{\useless_2},\lambda_{\useless_2})
\]
and
\[
W'_{v,*}=(\lambda_{\useless_1}, 1-\lambda_{\useless_1}).
\]
By Proposition~\ref{useless}, the maximum of
$\rankProbability_\lambda(1)=\lambda_{\useless_1}
 +\lambda_{\useless_2}$
is $\alpha_1 + \alpha_2$ with each $\alpha_j$ being the maximum of
feasible values of $\lambda_{\useless_j}$, so that
\[
\lowerIC_{10}(W)
= \capacity\left(\begin{pmatrix}
 1-\alpha_2 &\alpha_2\\
 \alpha_1 &1-\alpha_1
 \end{pmatrix}\right).
\]
Observing that these two rows are exactly those of $W$, we further
have $\lowerIC_{10}(W)=\capacity(W)$.
Again by Proposition~\ref{useless}, the minimum $\lowerRP_W(1)$ of
$\rankProbability_\lambda(1)$ is $(g-m+1)_+$.
With no loss of generality, we suppose $g=(1W)_1$.
Then the minima of feasible values of $\lambda_{\useless_1}$ and
$\lambda_{\useless_2}$ are $(g-m+1)_+$ and $0$, respectively, so that
\[
\upperIC_{10}(W)
= \capacity\left(\begin{pmatrix}
 1 &0\\
 \lowerRP_W(1) &1-\lowerRP_W(1)
 \end{pmatrix}\right).\qedhere
\]
\end{proof}

The fact that $\lowerIC_{10}(W)=\capacity(W)$ for binary-output channels is quite intriguing (although it is not true in general when the output is non-binary
(Example~\ref{counterExampleOfIC10})).
It implies that every binary-output channel can be simulated in a certain way that the capacity cannot be increased even when the encoder has causal access to the source of randomness, i.e., the intrinsic state.
The following result shows that, in fact for a fairly broad class of channels, the causal state information at the encoder is useless as far as the capacity is concerned.


\begin{theorem}\label{stateInfoUseless}
Let $W=W'W''$, where $W'$ is a channel with binary output and $W''$ is
a channel with binary input and $W''_{1,*}\ne W''_{2,*}$.
Suppose
\[
W'=\sum_{s\in\mathcal{S}} p_S(s) K^{(s)}
\]
where $S$ denotes the channel state and $p_S$ is its distribution.
The capacity of $W$ cannot be increased by the causal state
information $S$ at the encoder iff all $K^{(s)}$ with $p_S(s)>0$ are
$(i_1,i_2)$-ended for some fixed $i_1$ and $i_2$, where a binary
output channel $K$ is said to be $(i_1,i_2)$-ended if
$K_{i_1,1}=\min_{i}K_{i,1}$ and $K_{i_2,1}=\max_{i}K_{i,1}$.
In other words, all row vectors of $K$ are contained in the line
segment from endpoint $K_{i_1,*}$ to endpoint $K_{i_2,*}$.
\end{theorem}

\begin{proof}
(Sufficiency)
By \cite[Theorem~7.2 and Remark~7.6]{el_gamal_network_2011}, we
consider the channel
$\map{V}{\fromTo{1}{m}^\mathcal{S}}{\fromTo{1}{n}}$ given by
$V=V'W''$ and
\[
V'_{u,*} = \sum_{s\in\mathcal{S}} p_S(s) K^{(s)}_{u(s),*}.
\]
Because every channel $K^{(s)}$ is $(i_1,i_2)$-ended, it is easy to
show that $V'$ is also $(i_1,i_2)$-ended, where $i_1$ and $i_2$ are
regarded as two constant maps from $\mathcal{S}$ to
$\fromTo{1}{m}$.
Then every row vector of $V$ is contained in the line segment between
$V'_{i_1,*}W''$ and $V'_{i_2,*}W''$, which implies that  $V$ has a
capacity-achieving input probability distribution supported on
$\{i_1,i_2\}$ (Proposition~\ref{capacityAchieving.1}), and consequently
the capacity of $W$ cannot be increased by the causal state
information at the encoder.

(Necessity) If the capacity of $W$ cannot be increased by its causal
state information at the encoder, then a capacity-achieving input
probability distribution of $V$ must have a support, say
$\{i_1,i_2\}$, so that for every map
$\map{u}{\mathcal{S}}{\fromTo{1}{m}}$, the vector
\[
V_{u,*}
= V'_{u,*}W''
= \left(\sum_{s\in\mathcal{S}} p_S(s) K^{(s)}_{u(s),*}\right)W''
\]
is contained in the line segment between $V_{i_1,*}$ and $V_{i_2,*}$
(Proposition~\ref{capacityAchieving.2}), where $i_1$ and $i_2$ are
understood as two constant maps from $\mathcal{S}$ to $\fromTo{1}{m}$.
With no loss of generality, we assume $V'_{i_1,1}\le V'_{i_2,1}$.
For any $t\in\mathcal{S}$ and any $i_0\in\fromTo{1}{m}$, we
can take $u(t)=i_0$ and $u(s)=i_1$ for $s\ne t$, then we get
$V'_{u,1}\ge V'_{i_1,1}$, so that
$K^{(t)}_{i_0,1}\ge K^{(t)}_{i_1,1}$.
Similarly, we have $K^{(t)}_{i_0,1}\le K^{(t)}_{i_2,1}$.
Therefore, every $K^{(s)}$ is $(i_1,i_2)$-ended.
\end{proof}

It can be shown via a perturbation and continuity argument that the uselessness of the causal state information at the encoder is not restricted to the channels covered by
Theorem~\ref{stateInfoUseless}. However, we
have not been able to identify a simple explicit condition under which the sufficiency part of Theorem~\ref{stateInfoUseless}
can be extended. For example, consider a seemingly natural condition postulated by the following conjecture.


\begin{conjecture}
Let $W$ be a channel from $\fromTo{1}{2}$ to $\fromTo{1}{n}$.
Suppose
\[
W=\sum_{s\in\mathcal{S}} p_S(s) K^{(s)},
\]
where $S$ denotes the state of channel.
If for every $1\le j\le n$, $K^{(s)}_{1,j}$ and $K^{(s)}_{2,j}$ have
an order (either $\le$ or $\ge$) independent of $s$, then the capacity
of $W$ cannot be increased by the causal state information available
at the encoder.
\end{conjecture}

This conjecture is obviously true for $n=2$. Numerical results indicate that it also holds in many cases when
$n>2$. However it turns out to be false in general as shown by
Example~\ref{counterExampleOfIC10Conjecture}.

Theorem~\ref{stateInfoUseless} imposes no restriction on the distribution of the channel state. This universal property motivates us to introduce the following definition.

\begin{definition}\label{def:uslessE}
The state information $S$ of a channel $W(y\mid x,s)$ is said to be
\emph{universally useless at the encoder} if for any $p_S$, the
capacity of $W$ with $S$ causally available at the encoder is equal to
the capacity of $W'(y\mid x)=\sum_s p_S(s)W(y\mid x,s)$.
\end{definition}

This definition is not void in view of Theorem~\ref{stateInfoUseless} (in fact, according to our numerical results, many channels not covered by Theorem~\ref{stateInfoUseless} also satisfy this definition). Now consider the channel model shown in Fig. \ref{fig:modelgnou}, where the channel state $S$ is distributed according to $p_S$, and (noisy) state observations $S_\encoder$ and $S_\decoder$ generated by $S$ through $p_{S_\encoder,S_\decoder|S}$ are causally available  at the encoder and the decoder, respectively. Let $\capacity(W,S_\encoder,S_\decoder,p_{S})$ denote the capacity of this channel model.




\begin{figure}[h!]
	\centering
	\includegraphics[width=10cm]{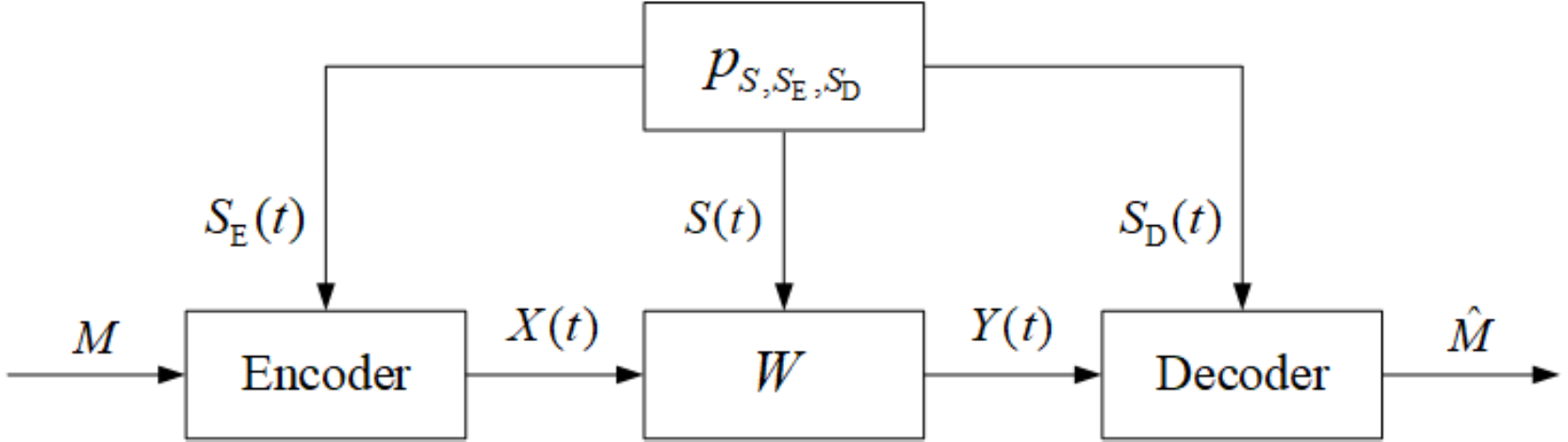}
	\caption{A generalized channel model.}
	\label{fig:modelgnou}
\end{figure}

It is instructive to study the following example (see also Fig. \ref{fig: channel with state}) where
\begin{align}
&W(y\mid x,s)=\left\{
\begin{array}{ll}
\frac{1}{2}, & (x,y,s)=(0,0,0)\mbox{, }(0,1,0)\mbox{, }(1,0,1)\mbox{ or }(1,1,1), \\
0, & (x,y,s)=(1,0,0)\mbox{ or }(0,1,1),\\
1, & (x,y,s)=(1,1,0)\mbox{ or }(0,0,1),
\end{array}
\right.\label{eq:para1}\\
&p_S(0)=p_S(1)=\frac{1}{2}.\label{eq:para2}
\end{align}
For this example, we assume that $p_{S_\encoder|S}$ is a binary symmetric channel with crossover probability $p\in[0,\frac{1}{2}]$, and $p_{S_\decoder|S}$ is a binary symmetric channel with crossover probability $q=0.25$; furthermore, we assume that $p_{S_\encoder|S}$ is physically degraded with respect to $p_{S_\decoder|S}$ when $p\geq q=0.25$, and the other way around when $p\leq q=0.25$. To gain a better understanding, we plot $\capacity(W,S_\encoder,S_\decoder,p_{S})$ against $p$ for $p\in[0,\frac{1}{2}]$ in Fig. \ref{fig: cap}. It turns out that, somewhat counterintuitively, $\capacity(W,S_\encoder,S_\decoder,p_{S})$ is maximized when the encoder side information coincides with the decoder side information (i.e., $p=0.25$) rather than when the encoder has access to the perfect state information $S$ (i.e., $p=0$). As shown by the following theorem, this is in fact a general phenomenon for any channel whose state information is universally useless at the encoder.

\begin{figure}[h!]
	\centering
	\includegraphics[width=10cm]{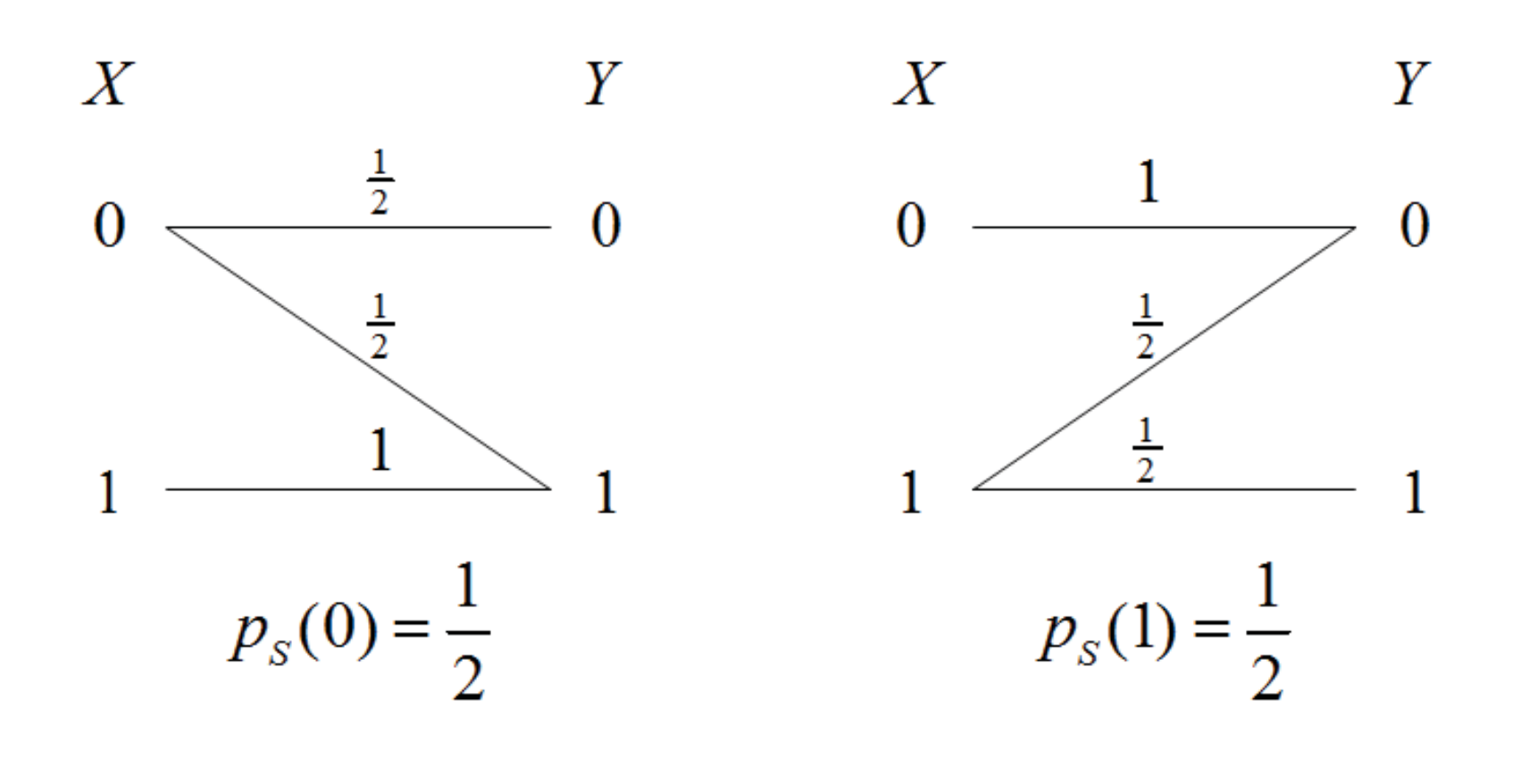}
	\caption{Illustration of $W$ and $p_S$ given by (\ref{eq:para1}) and (\ref{eq:para2}), respectively.}
	\label{fig: channel with state}
\end{figure}

\begin{figure}[h!]
	\centering
	\includegraphics[width=10cm]{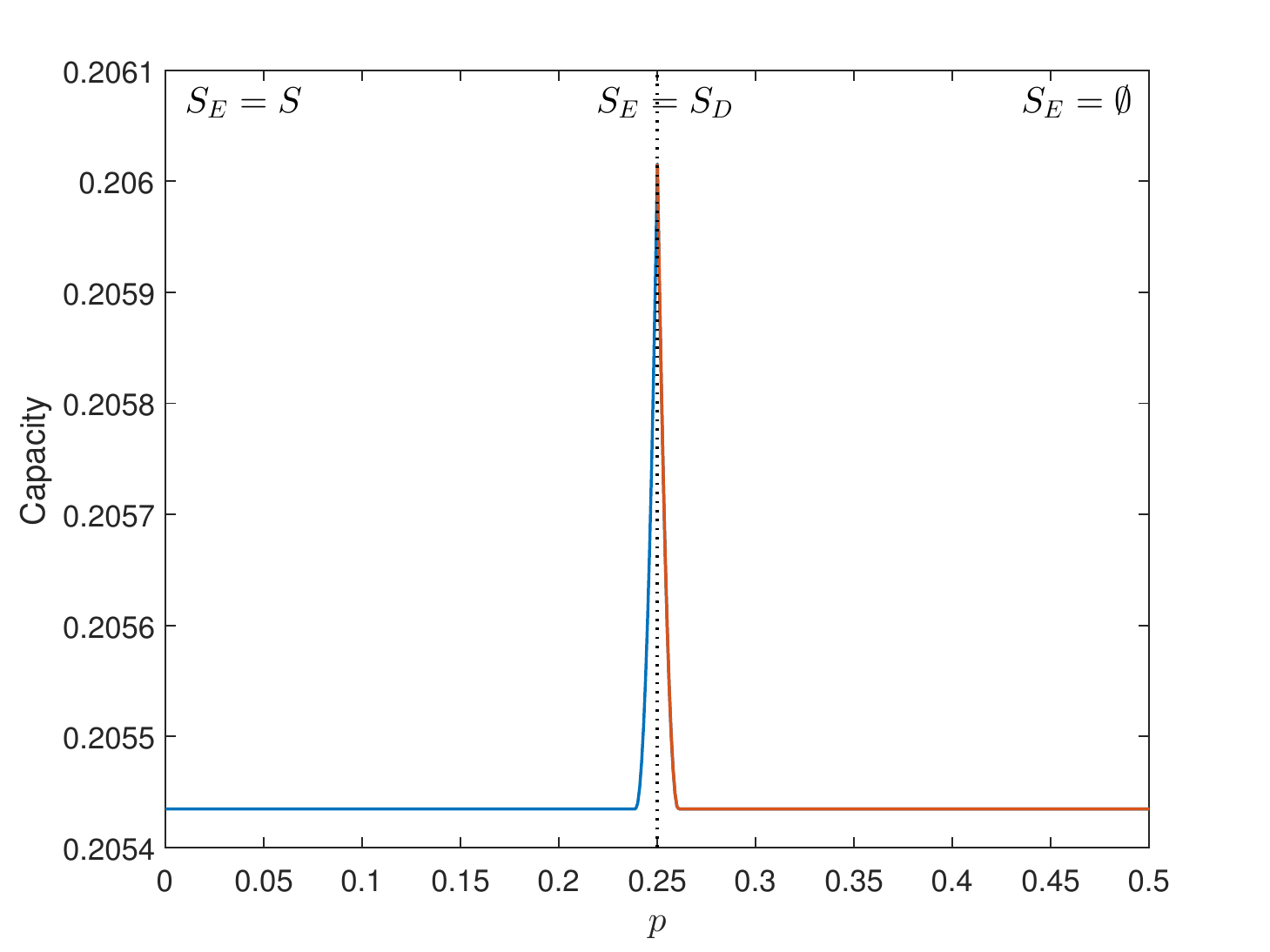}
	\caption{Plot of $\capacity(W,S_\encoder,S_\decoder,p_{S})$ against $p$ for $p\in[0,0.5]$, where $W$ and $p_S$ are given by (\ref{eq:para1}) and (\ref{eq:para2}), respectively.}
	\label{fig: cap}
\end{figure}


\begin{theorem}\label{maxCap}
If the state information of $W$ is universally useless at the encoder,
then $\capacity(W,S_\encoder,S_\decoder,p_{S})$ is maximized when $S_\encoder=S_\decoder$ almost surely (assuming $p_{S,S_\decoder}$ is fixed but $p_{S_\encoder|S,S_\decoder}$ can be arbitrary).
\end{theorem}

\begin{proof}
It is clear that among all possible forms of encoder side information $S_\encoder$, $\capacity(W,S_\encoder,S_\decoder,p_{S})$ is maximized
when $S_E=(S,S_\decoder)$ (since any other form of $S_E$ can be viewed as its degenerate version), i.e.,
\[
\capacity(W,S_\encoder,S_\decoder,p_{S})
\le \capacity(W,(S,S_\decoder),S_\decoder,p_{S}).
\]
Note that
\begin{eqnarray*}
\capacity(W,(S,S_\decoder),S_\decoder,p_{S})
&= &\sum_{s_\decoder} p_{S_\decoder}(s_\decoder)
 \capacity(W,S,\emptyset,p_{S\mid S_\decoder=s_\decoder})\\
&\eqvar{(a)} &\sum_{s_\decoder} p_{S_\decoder}(s_\decoder)
 \capacity(W,\emptyset,\emptyset,p_{S\mid S_\decoder=s_\decoder})\\
&= &\capacity(W,S_\decoder,S_\decoder,p_{S}),
\end{eqnarray*}
where (a) follows from the universal-uselessness property of the state
information of $W$, and the constant $\emptyset$ means no information. This completes the proof.
\end{proof}

Roughly speaking, Theorem~\ref{maxCap} implies that, for the class of channels satisfying Definition \ref{def:uslessE}, what the encoder really needs to know is not the state information, but the decoder's knowledge of the state information; in other words, for such channels, it is important to maintain consensus between the encoder and the decoder. It is also worth noting that Theorem~\ref{maxCap} reduces to Definition \ref{def:uslessE} when there is no decoder side information.


Another surprising phenomenon revealed by Fig. \ref{fig: cap} is that, as $p$ moves away from $0.25$, the capacity not only decreases but actually drops to the value corresponding to the no encoder side information case once $p$ passes certain thresholds. Again, such a phenomenon is not confined to that specific example. An investigation of this phenomenon in the context where the encoder side information is a degenerate version of the decoder side information can be found in \cite{xu_when_2017}.




Similar to Theorem~\ref{IC10}, we can also determine the exact values
of $\lowerIC_{01}(W)$ and $\upperIC_{01}(W)$ when the input is binary.
In this case, we have
$\capacity_{01}(\lambda)=\capacity_{11}(\lambda)$ for all
$\lambda\in\decompose(W)$, so that $\lowerIC_{01}(W)=\lowerIC_{11}(W)$
and $\upperIC_{01}(W)=\upperIC_{11}(W)$ (see Proposition~\ref{IC01}
and Appendix~\ref{resultIC01}).
The general case of $\lowerIC_{01}(W)$ and $\upperIC_{01}(W)$ is
however quite difficult.
Currently, we only know that $\upperIC_{01}(W)=\upperIC_{11}(W)$ does not
hold in general (Proposition~\ref{counterExampleOfIC01}).

\section{Conclusion}

We have studied the lower and the upper intrinsic capacities of a channel $W$, denoted by $\lowerIC_f(W)$ and
$\upperIC_f(W)$,  for three different scenarios ($f=10,01,11$) in terms of the availability of the causal state
information at the encoder and/or the decoder. Their values are determined in almost all cases when the input or the
output are binary, with only two exceptions (which are the binary-input
nonbinary-output channels for $f=10$ and the nonbinary-input
binary-output channels for $f=01$). A deeper understanding of the relevant optimization problems (especially the structure of $\decompose(W)$) is needed for further progress.

The lower and the upper intrinsic capacities are inherent properties of a channel with clear operational meanings. In particular, they characterize the potential capacity gains that can be achieved with a direct access to the generator of channel randomness by the encoder and/or the decoder. More generally, the notion of intrinsic capacity provides a useful perspective for studying the values of encoder and decoder side information. For example, our analysis of $\lowerIC_{10}(W)$ reveals that for a broad class of channels, the capacity is not necessarily maximized when the encoder has access to the perfect state information. We believe that this surprising finding is just the tip of the iceberg, and this line of research can be fruitfully pursued to uncover many previously unknown phenomena.

\appendix

\section{The Structure of $\decompose(W)$}
\label{structureOfDec}

\begin{theorem}\label{convexDec}
The set $\decompose(W)$ is a bounded, closed convex polytope.
For each $f\in\{10,01,11\}$, $\intrinsic_f(W)$ is a closed interval
and $\upperIC_f(W)$ can be attained at some vertex of $\decompose(W)$.
Furthermore, $\lowerIC_{11}(W)$ can also be attained at some vertex of
$\decompose(W)$.
\end{theorem}

\begin{proof}
By definition, it is clear that $\decompose(W)$ is a bounded, closed
convex polytope, so that $\intrinsic_f(W)$ is a closed interval
(Proposition~\ref{continuityOfC}).
It is also easy to see that $\capacity_f(\lambda)$ attains its maximum
$\upperIC_f(W)$ at some vertex of $\decompose(W)$ and that
$\capacity_{11}(\lambda)$ attains its minimum $\lowerIC_{11}(W)$ at
some vertex of $\decompose(W)$ (Proposition~\ref{continuityOfC} and
\cite[Proposition~3.4.1]{bertsekas_convex_2003}).
\end{proof}

In light of Theorem~\ref{convexDec}, we proceed to study the
structure of $\decompose(W)$ with a focus on its vertices.
Our approach is analogous to \cite{jurkat_extremal_1968}.

\begin{proposition}\label{decVertex1}
Let
\[
\mathfrak{S}
= \set{\support(\alpha)}{\alpha\in\real^\deterministic,
 \sum_{D\in\deterministic} \alpha_D D = 0}
\]
or
\[
\set{\support(\alpha)}{\alpha\in\real^\deterministic, \alpha I = 0},
\]
where
\begin{equation}\label{incidenceMatrix}
I
\eqdef (I_{D,(i,j)})_{D\in\deterministic,
 (i,j)\in\fromTo{1}{m}\times\fromTo{1}{n}}
= (D_{i,j})_{D\in\deterministic,
 (i,j)\in\fromTo{1}{m}\times\fromTo{1}{n}}
\end{equation}
is called the incidence matrix.
A probability distribution $\lambda\in\decompose(W)$ is a vertex iff
for $S\in\mathfrak{S}$, $S\subseteq\support(\lambda)$ implies
$S=\emptyset$, or in other words, iff $\rank(I_{S,*})=\size{S}$.
\end{proposition}

\begin{proof}
Note that for every $i\in\fromTo{1}{m}$,
\begin{equation}\label{hiddenEquation}
\sum_{D\in\deterministic} \alpha_D
= \sum_{j=1}^n \sum_{D\in\deterministic} \alpha_D D_{i,j}.
\end{equation}

(Sufficiency) If $\lambda=t\beta+(1-t)\gamma$ for some
$\beta,\gamma\in\decompose(W)$ and some $0<t<1$, then
$\beta-\gamma=(\lambda-\gamma)/t$ and
$\support(\gamma)\subseteq\support(\lambda)$, so that
$\support(\beta-\gamma)\in\mathfrak{S}$ and
$\support(\beta-\gamma)\subseteq\support(\lambda)$, hence
$\support(\beta-\gamma)=\emptyset$,
and therefore $\lambda=\beta=\gamma$ is a vertex.

(Necessity) For every nonempty $S\in\mathfrak{S}$, there is a vector
$\alpha\in\real^\deterministic$ such that $\support(\alpha)=S$ and
$\sum_D\alpha_DD=0$.
Let $\beta=\lambda+t\alpha$ and $\gamma=\lambda-t\alpha$ with
$t\ne 0$, so that $\lambda=(\beta+\gamma)/2$ with $\beta\ne\gamma$.
Since $\lambda$ is a vertex, $\beta$ and $\gamma$ must not be elements
of $\decompose(W)$ for all $t\ne 0$, or equivalently,
$S\not\subseteq\support(\lambda)$.
\end{proof}

Below are several easy consequences of Proposition~\ref{decVertex1}.

\begin{proposition}\label{decVertex2}
Let
\[
\mathfrak{T}
= \set{\support(\alpha)}{\alpha\in\decompose(W)}.
\]
A probability distribution $\lambda\in\decompose(W)$ is a vertex iff
$\support(\lambda)$ is minimal in $\mathfrak{T}$, where a minimal
pattern in $\mathfrak{T}$ is a set $T\subseteq\deterministic$ such
that $T=\support(\alpha)$ for some $\alpha\in\decompose(W)$ and for
every $\beta\in\decompose(W)$, $\support(\beta)\subseteq T$ implies
$\beta=\alpha$.
\end{proposition}

\begin{proposition}\label{decVertex3}
If $\lambda\in\decompose(W)$ is a vertex, then
\[
\weight(\lambda)\le \weight(W)-m+1.
\]
\end{proposition}

\begin{proof}[Sketch of Proof]
Because of \eqref{hiddenEquation}, the equations $\alpha I=0$ have at
most $m(n-1)+1$ linearly independent equations.
This number can be further reduced to $\weight(W)-m+1$ by utilizing
the information of $W$, because all the variables $\alpha_D$ with
$D_{i,j}=1$ must be zero if the equation
$\sum_{D\in\deterministic} \alpha_D D_{i,j}=W_{i,j}=0$.
The remaining part of the proof is then straightforward.
\end{proof}

Proposition~\ref{decVertex3} provides an upper bound for the support
size of a vertex in $\decompose(W)$.
On the other hand, the following result provides a lower bound for the
support size of points in $\decompose(W)$, including all the vertices
of $\decompose(W)$.

\begin{proposition}\label{decVertex3aux}
For any $\lambda\in\decompose(W)$,
\[
\weight(\lambda)
\ge \ceil{\log_2 s} \vee \weight(W_{1,*}) \vee \cdots
 \vee \weight(W_{m,*}),
\]
where
$s=\size{\{W_{i,j}\}_{i\in\fromTo{1}{m},j\in\fromTo{1}{n}}}$.
\end{proposition}

\begin{proof}
By the definition of $\decompose(W)$, we have
\[
W_{i,j}=\sum_{D\in\deterministic}\lambda_DD_{i,j}.
\]
Since $D_{i,j}$ is either $0$ or $1$, the right-hand side can yield at
most $2^{\weight(\lambda)}$ different values, so that
\[
2^{\weight(\lambda)}
\ge s
= \size{\{W_{i,j}\}_{i\in\fromTo{1}{m},j\in\fromTo{1}{n}}}.
\]
or $\weight(\lambda)\ge \ceil{\log_2 s}$.

On the other hand, every equation
\[
\sum_{D\in\deterministic}\lambda_DD_{i,j}=W_{i,j}>0
\]
must have at least one positive $\lambda_D$ for some
\[
D\in\mathfrak{D}_{i,j}=\set{D\in\deterministic}{D_{i,j}=1}.
\]
Since for every $i$, the sets $\mathfrak{D}_{i,1}$,
$\mathfrak{D}_{i,2}$, \ldots, $\mathfrak{D}_{i,n}$ are mutually
disjoint, we conclude that $\weight(\lambda)\ge\weight(W_{i,*})$.
\end{proof}

\begin{algorithm}\label{vertexAlgorithm}
Let $f$ be an arbitrary one-to-one map of $\fromTo{1}{n^m}$ onto
$\deterministic$.
The following algorithm with $W$ and $f$ as arguments can yield a
vertex of $\decompose(W)$.
\textnormal{
\begin{algorithmic}
\Function{vertex}{$W,f$}
 \State $\lambda\gets 0$,
  $K\gets W$, $i\gets 1$
 \While {$K\ne 0$ and $1\le i\le n^m$}
  \State $D\gets f(i)$
  \State $\lambda_D\gets\min_{1\le r\le m} K_{r,D(r)}$
  \State $K\gets K-\lambda_DD$
  \State $i\gets i+1$
 \EndWhile
  \State \Return $\lambda$
\EndFunction
\end{algorithmic}
}
\end{algorithm}

\begin{proof}[Sketch of Proof]
Let $\lambda$ be the vertex output by the algorithm.
Let $S=\support(\lambda)$.
Then by checking Algorithm~\ref{vertexAlgorithm}, it is easy to
verify that for every $D\in S$, there exists an $i\in\fromTo{1}{m}$
such that $I_{D,(i,D(i))}=1$ and $I_{D',(i,D(i))}=0$ for all
$D'\in S$ with $f^{-1}(D')>f^{-1}(D)$, so that
$\rank(I_{S,*})=\size{S}$.
\end{proof}

\begin{remark}
We can replace the map $f$ in Algorithm~\ref{vertexAlgorithm} with some
one-to-one map $\map{f'}{\fromTo{1}{\ell}}{\deterministic}$, where
$1\le\ell<n^m$.
Then we have a modified algorithm returning a pair $(\lambda', K)$
such that
\[
W = K + \sum_{D\in\deterministic}\lambda'_D D.
\]
Suppose the nontrivial case $K\ne 0$, so that
$\alpha=\sum_{D\in\deterministic}\lambda'_D<1$.
Let $W'=K/(1-\alpha)$.
If we have another algorithm to find a vertex of $\decompose(W')$,
say $\lambda''$, then it is easy to show that
$\lambda=\lambda'+(1-\alpha)\lambda''$ is a vertex of $\decompose(W)$.
\end{remark}

\section{Properties of $J_f$ and $C_f$}

This section provides some basic results on the analytic properties of
$J_f$ and $C_f$ defined in Section~\ref{problem}.
For any $p,p'\in\probability_A$,
\[
\statDistance(p,p')
\eqdef \frac{1}{2}\ellone{p-p'}
= \frac{1}{2}\sum_{a\in A}|p_a-p'_a|
\]
is called the \emph{statistical distance} on
$\probability_A$.
Given the product space $(A,\metric_A)\times(B,\metric_B)$, we
define its product metric by
\[
\productMetric((p,q),(p',q'))
\eqdef \metric_A(p,p')\vee\metric_B(q,q'),
\]
which induces the usual product topology.
Thus for any channels $W,W'\in\channel_{A,B}$, we have the
\emph{channel distance}
\[
\channelDistance(W,W')
\eqdef \productMetric((W_{a,*})_{a\in A},(W'_{a,*})_{a\in A})
= \max_{a\in A} \statDistance(W_{a,*},W'_{a,*}).
\]

\begin{proposition}\label{propertyOfJ}
(a) $J_{10}(\lambda,\mu)$ is uniformly continuous, and it is convex in
$\lambda$ for fixed $\mu$ and is concave in $\mu$ for fixed $\lambda$.

(b) $J_{01}(\lambda,\mu)$ is uniformly continuous, and it is linear in
$\lambda$ for fixed $\mu$ and is concave in $\mu$ for fixed $\lambda$.
\end{proposition}

\begin{proof}
(a) The function $J_{10}(\lambda,\mu)$ can be rewritten as
$I(\mu,g(\lambda))$ where
\[
g(\lambda)
= ((\lambda V(u))_y)_{u\in\mathcal{X}^\deterministic,y\in\mathcal{Y}}
\]
with $V(u)=(D_{u(D),y})_{D\in\deterministic,y\in\mathcal{Y}}$.
By Proposition~\ref{statDistance.P1}, for
$\lambda,\lambda'\in\probability_\deterministic$,
\[
\channelDistance(g(\lambda),g(\lambda'))
= \max_{u\in\mathcal{X}^\deterministic}
 \statDistance(\lambda V(u),\lambda'V(u))
\le \statDistance(\lambda,\lambda'),
\]
so that $g$ is uniformly continuous, and hence $J_{10}(\lambda,\mu)$
is uniformly continuous (Proposition~\ref{continuityOfI3}).
It is also clear that $g$ is linear, so that $J_{10}(\lambda,\mu)$ is
convex for fixed $\mu$ and is concave for fixed $\lambda$
(\cite[Theorem~2.7.4]{cover_elements_2006}).

(b) The function $J_{01}(\lambda,\mu)$ can be written as
$\lambda\transpose{(g(\mu))}$ where
$g(\mu)=(I(\mu,D))_{D\in\deterministic}$.
By Propositions~\ref{continuityOfI} and \ref{statDistance.P1},
$I(\mu,D)$ is uniformly continuous on $\probability_\mathcal{X}$ and
is bounded by $\log(|\mathcal{X}|\wedge|\mathcal{Y}|)$.
Then for $\lambda,\lambda'\in\probability_\deterministic$ and
$\mu,\mu'\in\probability_\mathcal{X}$, we have
\begin{eqnarray*}
& &|\lambda\transpose{(g(\mu))}-\lambda'\transpose{(g(\mu'))}|\\
& &= |\lambda\transpose{(g(\mu))}-\lambda'\transpose{(g(\mu))}
 + \lambda'\transpose{(g(\mu))}-\lambda'\transpose{(g(\mu'))}|\\
& &\le |\lambda\transpose{(g(\mu))}-\lambda'\transpose{(g(\mu))}|
 + |\lambda'\transpose{(g(\mu))}-\lambda'\transpose{(g(\mu'))}|\\
& &\le |\lambda-\lambda'|\transpose{(g(\mu))}
 + \lambda'\transpose{|g(\mu)-g(\mu')|}\\
& &\le \log(|\mathcal{X}|\wedge|\mathcal{Y}|)\ellone{\lambda-\lambda'}
 + \ellone{g(\mu)-g(\mu')}
\end{eqnarray*}
which implies that $J_{01}$ is uniformly continuous.
The remaining part is straightforward
(\cite[Theorem~2.7.4]{cover_elements_2006}).
\end{proof}

\begin{proposition}\label{continuityOfC}
For $f\in\{10,01,11\}$, $\capacity_f(\lambda)$ is uniformly
continuous and convex (and in fact linear for $f=11$).
\end{proposition}

\begin{proof}[Sketch of Proof]
Use Theorem~\ref{propertyOfJ} and Proposition~\ref{continuityOfSup}
for $f=10$ or $01$.
The case of $f=11$ is trivial because $\capacity_{11}(\lambda)$ is a
linear function of $\lambda$.
\end{proof}

\begin{proposition}[{\cite[Theorem~2]{zhang_estimating_2007}}]%
\label{continuityOfI}
For $\mu,\mu'\in\probability_A$ and $W,W'\in\channel_{A,B}$,
\[
\left|I(\mu,W)-I(\mu',W')\right|
\le 3\delta\log(|A||B|-1) + 3h(\delta),
\]
where $\delta = \statDistance(\diagonal(\mu)W,\diagonal(\mu')W')$.
\end{proposition}

\begin{proposition}%
[cf.\ {\cite[Lemma~3]{yassaee_achievability_2014}}]%
\label{statDistance.P1}
For $\mu,\mu'\in\probability_A$ and $W\in\channel_{A,B}$,
\[
\statDistance(\diagonal(\mu)W,\diagonal(\mu')W)
= \statDistance(\mu,\mu')
\]
and
\[
\statDistance(\mu W,\mu'W)
\le \statDistance(\mu,\mu').
\]
\end{proposition}

\begin{proposition}%
[cf.\ {\cite[Lemma~3]{yassaee_achievability_2014}}]%
\label{continuityOfI2}
For $\mu,\mu'\in\probability_A$ and $W,W'\in\channel_{A,B}$,
\begin{eqnarray*}
\statDistance(\diagonal(\mu)W,\diagonal(\mu')W')
&\le &\statDistance(\mu,\mu')+\channelDistance(W,W')\\
&\le &2\productMetric((\mu,W),(\mu',W')),
\end{eqnarray*}
so that $I(\mu,W)$ is uniformly continuous on
$(\probability_A\times\channel_{A,B},\productMetric)$.
\end{proposition}

\begin{proof}[Sketch of Proof]
Use the triangle inequality and
Propositions~\ref{continuityOfI}, \ref{statDistance.P1} and
\ref{statDistance.P2}.
\end{proof}

\begin{proposition}\label{continuityOfI3}
Let $g$ be a map from $\probability_C$ to $\channel_{A,B}$.
If $g$ is uniformly continuous, then $I(\mu,g(\lambda))$ is uniformly
continuous on $(\probability_A\times\probability_C,\productMetric)$,
where $\mu\in\probability_A$ and $\lambda\in\probability_C$.
\end{proposition}

\begin{proof}[Sketch of Proof]
Use Propsoitions~\ref{continuityOfI} and \ref{continuityOfI2} and the
observation that $I(\mu,g(\lambda))$ is a composition of uniformly
continuous maps.
\end{proof}

\begin{proposition}\label{continuityOfSup}
If $\map{g}{A\times B}{\real}$ is uniformly continuous on
$(A\times B,\productMetric)$, then $f(x)=\sup_{b\in B}g(x,b)$ is
uniformly continuous.
\end{proposition}

\begin{proof}
Since $g$ is uniformly continuous, for any $\epsilon>0$, there is a
$\delta>0$ such that for any $a,a'\in A$ and any $b\in B$,
$\productMetric((a,b),(a',b))<\delta$ implies
$|g(a,b)-g(a',b)|<\epsilon$.
In other words, for any $b\in B$, $\metric_A(a,a')<\delta$ implies
$|g(a,b)-g(a',b)|<\epsilon$.
Then
\[
\sup_{b\in B}g(a,b) - \sup_{b\in B}g(a',b)
\le \sup_{b\in B}(g(a,b)-g(a',b))
< \epsilon
\]
and similarly, $\sup_{b\in B}g(a',b)-\sup_{b\in B}g(a,b)<\epsilon$,
so that $f(x)$ is uniformly continuous.
\end{proof}

\begin{proposition}\label{statDistance.P2}
For $\mu\in\probability_A$ and $W,W'\in\channel_{A,B}$,
\[
\statDistance(\diagonal(\mu)W,\diagonal(\mu)W')
\le \channelDistance(W,W').
\]
\end{proposition}

\begin{proof}
\begin{eqnarray*}
\statDistance(\diagonal(\mu)W,\diagonal(\mu)W')
&= &\frac{1}{2}\sum_{a,b}|\mu_aW_{a,b}-\mu_aW'_{a,b}|\\
&= &\frac{1}{2}\sum_a\mu_a\sum_b|W_{a,b}-W'_{a,b}|\\
&= &\sum_a\mu_a\statDistance(W_{a,*},W'_{a,*})
\le \statDistance(W,W').
\end{eqnarray*}
\end{proof}

\section{Proofs and Examples of Section~\ref{secIC11}}
\label{secIC11extra}

\begin{proof}[Proof of Proposition~\ref{min.pattern.1}]
Let $W=(P_1+\cdots+P_\ell)/\ell$ and $j$ be the column such that
$\weight(W_{*,j})>1$.
It is clear that $W = xD + (1-x)W'$ for some $x\in (0,1)$ and some
$D\in\deterministic$ such that $\weight(D_{*,j})>1$, so that
$\lowerIC_{11}(W) < \log m$, and hence
$P_1$, \ldots, $P_\ell$ are not $\intrinsic$-minimized.
\end{proof}

\begin{proof}[Proof of Proposition~\ref{min.pattern.2}]
Let $W=(P_1+\cdots+P_\ell)/\ell$ and $j$ be the column of which all
entries are less than $1$.
It is clear that $W = xD + (1-x)W'$ for some $x\in (0,1)$ and some
$D\in\deterministic$ such that $\weight(D_{*,j})=0$, so that
$\lowerIC_{11}(W) < \log n$, and hence
$P_1$, \ldots, $P_\ell$ are not $\intrinsic$-minimized.
\end{proof}

\begin{proof}[Proof of Proposition~\ref{max.pattern.1}]
With no loss of generality, we assume that $D_{m,j}=0$.
It is then clear that
\[
\frac{1}{2}D + \frac{1}{2}\useless_j
= \frac{1}{2}D' + \frac{1}{2}D'',
\]
where
\[
D'(i)=\begin{cases}
j &if $i=m$,\\
D(i) &otherwise,
\end{cases}
\]
and
\[
D''(i)=\begin{cases}
D(m) &if $i=m$,\\
j &otherwise.
\end{cases}
\]
It is clear that $\rank(D')\ge (\rank(D)-1)\vee 2$ and $\rank(D'')=2$,
so that
\[
\frac{1}{2}\log\rank(D)+\frac{1}{2}\log\rank(\useless_j)
< \frac{1}{2}\log\rank(D')+\frac{1}{2}\log\rank(D''),
\]
and therefore $\{D, \useless_j\}$ is not $\intrinsic$-maximized.
\end{proof}

\begin{example}
If
\[
W
= \begin{pmatrix}
1/n &1/n &\cdots &1/n\\
1/n &1/n &\cdots &1/n\\
\vdots &\vdots &\ddots &\vdots\\
1/n &1/n &\cdots &1/n
\end{pmatrix},
\]
which is the probability transition matrix seen in the well-known random binning scheme,
then $\lowerIC_{11}(W)=0$ and $\upperIC_{11}(W)=\log(m\wedge n)$
(Theorems~\ref{lowerICBound} and \ref{upperICBound}).
\end{example}

\begin{example}\label{concrete}
\[
W = \begin{pmatrix}
0.3 &0.3 &0.4\\
0.2 &0.5 &0.3\\
0.4 &0.1 &0.5
\end{pmatrix}
\]
It can be computed using linear programming that
$\lowerIC_{11}(W)=0.4$ and $\upperIC_{11}(W)=0.2+0.8\log 3\approx 1.4680$.
The decompositions of $W$ for $\lowerIC_{11}(W)$ and $\upperIC_{11}(W)$ are
\begin{eqnarray*}
W
&= &0.2\begin{pmatrix}1 &0 &0\\ 1 &0 &0\\ 1 &0 &0\end{pmatrix}
 + 0.1\begin{pmatrix}0 &1 &0\\ 0 &1 &0\\ 0 &1 &0\end{pmatrix}
 + 0.3\begin{pmatrix}0 &0 &1\\ 0 &0 &1\\ 0 &0 &1\end{pmatrix}\\
& &+ 0.1\begin{pmatrix}1 &0 &0\\ 0 &1 &0\\ 1 &0 &0\end{pmatrix}
 + 0.1\begin{pmatrix}0 &1 &0\\ 0 &1 &0\\ 1 &0 &0\end{pmatrix}
 + 0.1\begin{pmatrix}0 &1 &0\\ 0 &1 &0\\ 0 &0 &1\end{pmatrix}
 + 0.1\begin{pmatrix}0 &0 &1\\ 0 &1 &0\\ 0 &0 &1\end{pmatrix}
\end{eqnarray*}
and
\begin{eqnarray*}
W
&= &0.1\begin{pmatrix}0 &0 &1\\ 1 &0 &0\\ 0 &0 &1\end{pmatrix}
 + 0.1\begin{pmatrix}0 &0 &1\\ 0 &1 &0\\ 0 &0 &1\end{pmatrix}\\
& &+ 0.1\begin{pmatrix}0 &0 &1\\ 0 &1 &0\\ 1 &0 &0\end{pmatrix}
 + 0.3\begin{pmatrix}0 &1 &0\\ 0 &0 &1\\ 1 &0 &0\end{pmatrix}
 + 0.1\begin{pmatrix}0 &0 &1\\ 1 &0 &0\\ 0 &1 &0\end{pmatrix}
 + 0.3\begin{pmatrix}1 &0 &0\\ 0 &1 &0\\ 0 &0 &1\end{pmatrix},
\end{eqnarray*}
respectively.
Using Theorems~\ref{lowerICBound} and \ref{upperICBound} and
Proposition~\ref{perfect}, we have
\[
0.4
\le \lowerIC_{11}(W)
\le 0.4\log 3
\approx 0.6340
\]
and
\[
1
\le \upperIC_{11}(W)
\le 0.1+0.9\log 3
\approx 1.5265.
\]
From Proposition~\ref{min.pattern.1}, it follows that the optimal
decomposition $\lambda'$ for $\lowerIC_{11}(W')$ can have at most one
perfect channel, so that $\rankProbability_{\lambda'}(3)\le 0.25$,
where
\[
W' = \begin{pmatrix}
0.25 &0.5 &0.25\\
0 &1 &0\\
0.5 &0 &0.5
\end{pmatrix}
\]
is computed by the formula in Theorem~\ref{lowerICBound}.
Then we have an improved bound:
$\lowerIC_{11}(W)\le 0.4\lowerIC_{11}(W')=0.3 + 0.1\log 3\approx 0.4585$.
\end{example}

\section{Capacity-Achieving Input Probability Distributions}

Let $W$ be a channel in $\channel_{\mathcal{X},\mathcal{Y}}$.
According to \cite[Theorem~4.5.1]{gallager_information_1968}, an input
probability distribution $\mu$ maximizes the mutual information
$I(\mu,W)$ iff
\[
D(W_{x,*}\|\tau)=C\quad\text{for $x\in\support(\mu)$}
\]
and
\[
D(W_{x,*}\|\tau)\le C\quad\text{for $x\notin\support(\mu)$},
\]
where $\tau=\mu W$.
Based on this sufficient and necessary condition, we have the
following results concerning the support of capacity-achieving input
probability distributions.
In the sequel, we denote by $\convexHull(V)$ the convex hull of all
vectors in $V$.

\begin{proposition}\label{capacityAchieving.1}
Let $A\subseteq\mathcal{X}$.
If all row vectors of $W$ are contained in
$\convexHull(\{W_{x,*}\}_{x\in A})$, then there exists a
capacity-achieving probability distribution $\mu$ such that
$\support(\mu)\subseteq A$.
\end{proposition}

\begin{proof}
Let $\nu$ be a capacity-achieving probability distribution of the sub
matrix $W_{A,*}$.
Extending $\nu$ with zero values, we obtain a probability distribution
$\mu$ over $\mathcal{X}$.
It is clear that
\[
D(W_{x,*}\|\tau)=C\quad\text{for $x\in\support(\mu)$}
\]
and
\[
D(W_{x,*}\|\tau)\le C\quad\text{for $x\in A\setminus\support(\mu)$},
\]
where $\tau=\mu W=\nu W_{A,*}$.
It remains to show that
\[
D(W_{x,*}\|\tau)\le C\quad\text{for $x\notin A$},
\]
which is obvious, because
\[
D(W_{x,*}\|\tau)
= D\left(\sum_{a\in A}\alpha_aW_{a,*} \Big\| \tau\right)
\le \sum_{a\in A}\alpha_aD(W_{a,*}\|\tau)
\le C
\]
for some nonnegative coefficients $(\alpha_a)_{a\in A}$ with
$\sum_{a\in A}\alpha_a=1$.
\end{proof}

\begin{proposition}\label{capacityAchieving.2}
Let $\mu$ be a capacity-achieving probability distribution of $W$ and
let $A=\support(\mu)$.
For any $a\in A$ and any $b\notin A$,
\(
W_{a,*}
\notin \convexHull(\{W_{x,*}\}_{x\in A\cup\{b\}}\setminus\{W_{a,*}\})
\).
\end{proposition}

\begin{proof}
It is clear that $D(W_{x,*}\|\tau)=C$ for all $x\in A$, where
$\tau=\mu W$.
We first show that
$W_{a,*}\notin\convexHull(\{W_{x,*}\}_{x\in A}\setminus\{W_{a,*}\})$,
which corresponds to the case $W_{b,*}=W_{a,*}$.
If it is false, then
\[
W_{a,*} = \sum_{x\in A'} \alpha_xW_{x,*}
\]
where $A'=\set{x\in A}{W_{x,*}\ne W_{a,*}}$, $\alpha_x\ge 0$, and
$\sum_{x\in A'} \alpha_x = 1$.
It is clear that $\alpha_x<1$ for all $x\in A'$, so that
\[
D(W_{a,*}\|\tau) < \sum_{x\in A'} D(W_{x,*}\|\tau) = C,
\]
a contradiction.
Now suppose that
\[
W_{a,*}
\in\convexHull(\{W_{x,*}\}_{x\in A\cup\{b\}}\setminus\{W_{a,*}\})
\]
for some $b\notin A$ with $W_{b,*}\ne W_{a,*}$.
Let $A''=A'\cup \{b\}$.
Then
\[
W_{a,*} = \sum_{x\in A''} \alpha_xW_{x,*}
\]
where $\alpha_x\ge 0$ and $\sum_{x\in A''} \alpha_x = 1$.
It is clear that $0<\alpha_b<1$, and therefore
\begin{eqnarray*}
C
&= &D(W_{a,*}\|\tau)
< \alpha_bD(W_{b,*}\|\tau)
 + \sum_{x\in A''\setminus\{b\}} \alpha_x D(W_{x,*}\|\tau)\\
&= &\alpha_bD(W_{b,*}\|\tau)+(1-\alpha_b)C,
\end{eqnarray*}
so that $D(W_{b,*}\|\tau)>C$, which is absurd.
\end{proof}

\section{Counterexamples for Section~\ref{secIC10}}
\label{counterOfIC10}

\begin{example}\label{counterExampleOfIC10}
$\lowerIC_{10}(W)>\capacity(W)$ for
\[
W
= \begin{pmatrix}
0.8 &0.2 &0\\
0.6 &0.35 &0.05
\end{pmatrix}.
\]
\end{example}

\begin{proof}
Let $S=\set{D\in\deterministic}{D(1)\in\{1,2\}, D(2)=3}$.
It is then clear that, for every $\lambda\in\decompose(W)$,
\[
\sum_{D\in S}\lambda_D = 0.05.
\]
If we define the map $\map{u}{\deterministic}{\fromTo{1}{2}}$ by
\[
u(D)
= \begin{cases}
1 &$D\in S$,\\
2 &otherwise,
\end{cases}
\]
then the row vector
$v=(\sum_D \lambda_D D_{u(D),y})_{y\in\fromTo{1}{3}}$ is always on the
line segment $L$ with endpoints $(0.65,0.35,0)$ and $(0.6,0.4,0)$.

By numerical computation, we know that
\[
\divergence{W_{1,*}}{c})
= \divergence{W_{2,*}}{c}
\approx 0.03541501,
\]
where
\[
c
= \mu W
\approx (0.71339243,0.26495568,0.02165189)
\]
and
\[
\mu \approx (0.56696216,0.43303784)
\]
is the capacity-achieving input distribution of $W$.
Furthermore, it can be verified that all points $x$ of $L$ satisfy
\[
\divergence{x}{c}>0.0369.
\]
This implies that $\mu$, if extended to
$\fromTo{1}{2}^\deterministic$, cannot be a capacity-achieving
distribution (\cite[Theorem~4.5.1]{gallager_information_1968}).
In other words, for every $\lambda\in\decompose(W)$, the intrinsic
capacity $\capacity_{10}(\lambda)>\capacity(W)$, so that
$\lowerIC_{10}(W)>\capacity(W)$.
\end{proof}

\begin{example}\label{counterExampleOfIC10Conjecture}
Let state alphabet $\mathcal{S}=\fromTo{1}{2}$ and let
\begin{eqnarray*}
W
&= &\sum_{s\in\mathcal{S}} p_S(s) K^{(s)}
= \frac{17}{18}K^{(1)}+\frac{1}{18}K^{(2)}\\
&= &\begin{pmatrix}\delta\\ \gamma\end{pmatrix}
= \begin{pmatrix}0.05 &0.1 &0.85\\ 0 &0.05 &0.95\end{pmatrix},
\end{eqnarray*}
where
\[
K^{(1)}
=\begin{pmatrix}\alpha\\ \gamma\end{pmatrix}
=\begin{pmatrix}0 &0.1 &0.9\\ 0 &0.05 &0.95\end{pmatrix}
\]
and
\[
K^{(2)}
=\begin{pmatrix}\beta\\ \gamma\end{pmatrix}
=\begin{pmatrix}0.9 &0.1 &0\\ 0 &0.05 &0.95\end{pmatrix}.
\]
It is easy to show that $\mu=(0.603123,0.396877)$ is the
capacity-achieving input distribution for $W$, so that the output
distribution is
\[
\tau
\eqdef \mu_1\delta+\mu_2\gamma
\approx (0.01984385,0.06984385,0.9103123)
\]
and $D(\delta\|\tau)=D(\gamma\|\tau)\approx 0.0238286$.
However, for the channel
$\map{V}{\fromTo{1}{2}^\mathcal{S}}{\fromTo{1}{3}}$ given by
\[
V_{u,*} = \sum_{s\in\mathcal{S}} p_S(s) K^{(s)}_{u(s),*}
\qquad{\text{(\cite[Remark~7.6]{el_gamal_network_2011})}},
\]
if we choose the map $u(s)=s$, then the corresponding row vector
\begin{eqnarray*}
\zeta
&= &V_{u,*}
= p_1K^{(1)}_{1,*}+p_2K^{(2)}_{2,*}\\
&= &\frac{17}{18}\alpha+\frac{1}{18}\gamma
\approx (0,0.09722222,0.90277778)
\end{eqnarray*}
and $D(\zeta\|\tau)\approx 0.0246518 > D(\gamma\|\tau)$.
This implies that $\mu$, if extended to $\fromTo{1}{2}^\mathcal{S}$,
cannot be a capacity-achieving distribution for $V$
(\cite[Theorem~4.5.1]{gallager_information_1968}).
In other words, the capacity of $W$ can be increased by the causal
state information $S$ at the encoder.
\end{example}

\section{$\lowerIC_{01}(W)$ and $\upperIC_{01}(W)$}
\label{resultIC01}

\begin{proof}[Proof of Proposition~\ref{IC01}]
Because $m=2$, the binary uniform distribution is capacity-achieving
for every deterministic channel, rank $1$ or rank $2$.
Thus we have $\capacity_{01}(\lambda)=\capacity_{11}(\lambda)$ for
every $\lambda\in\decompose(W)$.
The remaining part is an easy consequence of
Propositions~\ref{lowerICBound} and \ref{upperICBound}.
\end{proof}

\begin{proposition}\label{counterExampleOfIC01}
Let $W$ be a channel $\fromTo{1}{3}\to\fromTo{1}{2}$.
If all probabilities $W_{i,j}$ are distinct and the sum of each column
of $W$ is greater than or equal to $1$, then
$\upperIC_{01}(W)<\upperIC_{11}(W)$.
\end{proposition}

\begin{proof}
By Proposition~\ref{perfect}, $\upperRP_W(2)=1$, so that $W$ can be
expressed as a convex combination of perfect channels and hence
$\upperIC_{11}(W)=1$.

Let
\[
S
= \set{\lambda\in\decompose(W)}{\rankProbability_\lambda(2)=1}.
\]
If $\upperIC_{01}(W)=1$, then there exists a $\lambda\in S$ such that
the capacity-achieving input distribution, denoted $\mu$, is
capacity-achieving for every perfect channel $D\in\support(\lambda)$.
Thus at least one entry of $\mu$ must be $1/2$.
With no loss of generality, we assume $\mu_1=1/2$.

If $\mu_2$ and $\mu_3$ are both positive, then $\mu$ is
capacity-achieving only for perfect channels
\[
\begin{pmatrix}
1 &0\\
0 &1\\
0 &1
\end{pmatrix}
\text{ and }
\begin{pmatrix}
0 &1\\
1 &0\\
1 &0
\end{pmatrix}.
\]
By Proposition~\ref{decVertex3aux}, every $\lambda\in\decompose(W)$
satisfies $\support(\lambda)\ge\ceil{\log_2 6}=3$, which implies that
$\mu$ is not capacity-achieving for $\lambda\in S$.

If $\mu_2=0$, then $\mu$ is capacity-achieving for perfect channels
\[
\begin{pmatrix}
1 &0\\
0 &1\\
0 &1
\end{pmatrix},
\begin{pmatrix}
0 &1\\
1 &0\\
1 &0
\end{pmatrix},
\begin{pmatrix}
1 &0\\
1 &0\\
0 &1
\end{pmatrix},
\begin{pmatrix}
0 &1\\
0 &1\\
1 &0
\end{pmatrix}.
\]
However, any convex combination of these four matrices can only yield
a channel matrix with at most four distinct probability values, and
hence $\mu$ is not capacity-achieving for $\lambda\in S$.

In all cases, we have shown that $\mu$ is not capacity-achieving,
which contradicts the assumption $\upperIC_{01}(W)=1$.
Therefore, we have $\upperIC_{01}(W)<1=\upperIC_{11}(W)$.
\end{proof}

\small
\bibliographystyle{IEEEtran}
\bibliography{IEEEabrv,ic}

\bigbreak
{%
\footnotesize\raggedleft

\texttt{(Version~\docVersion.\docBuildNumber)}

{\textsf{article} options: $\mathrm{\paperOption}$}\\
{\textsf{geometry} options: $\mathrm{\geometryOption}$}\\
{\textsf{hyperref} options: $\mathrm{\hyperrefOption}$}\par
}

\end{document}